\RequirePackage{ifpdf}
\ifpdf
    \documentclass[pdftex,envcountsect,envcountsame]{llncs}
\else
    \documentclass[dvips,envcountsect,envcountsame]{llncs}
\fi

\usepackage{xspace}
\usepackage{textgreek}
\usepackage{latexsym,amssymb,amsmath,amsfonts}
\usepackage{paralist}
\usepackage{todonotes}
\usepackage{tabularx}
\usepackage{enumerate}
\usepackage[latin1]{inputenc}
\usepackage[pdftitle   = {{Quantifier elimination for
database driven verification}},
            pdfauthor  = {{Calvanese -Ghilardi - Gianola - Montali - Rivkin}},
            pdfsubject = {{}}]{hyperref}
\ifpdf \usepackage{url} \else \usepackage{breakurl} \fi
\usepackage{comment}
\usepackage{graphicx}
\usepackage{tikz-cd}
\usepackage{wrapfig}
\usepackage{amsmath,amssymb}
\usepackage{mathptmx}
\DeclareMathAlphabet{\mathcal}{OMS}{cmsy}{m}{n}

\usepackage[linesnumberedhidden,ruled,vlined]{algorithm2e}
\usepackage{graphicx}
\title{Quantifier Elimination for
Database Driven Verification
}

\author{%
Diego Calvanese$^1$, Silvio Ghilardi$^2$,  Alessandro Gianola$^1$, \\ Marco Montali$^1$, Andrey Rivkin$^1$ 
}

\institute{%
 $^1$Faculty of Computer Science, Free University of Bozen-Bolzano (Italy)\\
\{calvanese, gianola, montali, rivkin\}@inf.unibz.it \\
 $^2$Dipartimento di Matematica, Universit\`a degli Studi di Milano (Italy)\\
 silvio.ghilardi@unimi.it
}

\newcommand{\I}{\ensuremath{\mathcal{I}}}

\newcommand{\M}{\ensuremath{\mathcal{M}}}

\newcommand{\set}[1]{\{#1\}}                      % set

\newcommand{\tup}[1]{\langle #1\rangle}            % tuple

\newcolumntype{C}{>{\centering\arraybackslash}X}

\makeatletter
\g@addto@macro\normalsize{%
\setlength{\abovecaptionskip}{-2pt}
\setlength\abovedisplayskip{3pt}
\setlength\belowdisplayskip{3pt}
\setlength\abovedisplayshortskip{3pt}
\setlength\belowdisplayshortskip{3pt}
}
\makeatother

\newcounter{dummy} 
\newcounter{dummy1} 
\newcounter{dummy2}
\newcounter{dummy3} 
\newcounter{dummy4}
\newcounter{dummy5} 
\newcounter{dummy6}
\usepackage[thmmarks,amsmath]{ntheorem}
\theorempreskip{1pt}
\theorempostskip{1pt}

\theoremseparator{.}
\theorembodyfont{\itshape}
\theoremsymbol{$\triangleleft$}
\newtheorem{theorem}[dummy]{Theorem}
\newtheorem{lemma}[dummy1]{Lemma}
\newtheorem{definition}[dummy2]{Definition}
\newtheorem{proposition}[dummy3]{Proposition}

\theorembodyfont{\normalfont}
\newtheorem{example}[dummy4]{Example}
\newtheorem{remark}[dummy5]{Remark}

\theoremstyle{nonumberplain}
\theoremheaderfont{\itshape}
\theorembodyfont{\normalfont}

\theoremseparator{.}
\theoremsymbol{\ensuremath{\dashv}}
\newtheorem{proof}[dummy6]{Proof}

\qedsymbol{\ensuremath{\dashv}}

\newcommand{\ua}{\ensuremath{\underline a}}
\newcommand{\ub}{\ensuremath{\underline b}}

\newcommand{\ue}{\ensuremath{\underline e}}

\newcommand{\ux}{\ensuremath{\underline x}}
\newcommand{\uy}{\ensuremath{\underline y}}
\newcommand{\uz}{\ensuremath{\underline z}}

\newcommand{\cM}{\ensuremath \mathcal M}
\newcommand{\cN}{\ensuremath \mathcal N}
\newcommand{\cS}{\ensuremath \mathcal S}

\newcommand{\cI}{\ensuremath \mathcal I}
\newcommand{\lra}{\longrightarrow}

\newcommand{\mcmt}{\textsc{mcmt}\xspace}

\renewcommand{\int}{\ensuremath {\mathcal I}}

\newcommand{\domain}[1]{\mathit{dom}(#1)}

\newcommand{\constant}[1]{\texttt{#1}}

\newcommand{\functs}[1]{#1_{\mathit{fun}}}
\newcommand{\vals}[1]{#1_{\mathit{val}}}
\newcommand{\ids}[1]{#1_{\mathit{ids}}}

\newcommand{\nullv}{\texttt{undef}}

\newcommand{\lentry}[2]{\ensuremath{#1\,{:}\,#2}}
\definecolor{deepblue}{HTML}{0C3B80}
\definecolor{deepgreen}{HTML}{2EA601}
\definecolor{lightOrange}{HTML}{FFA03C}
\definecolor{darkOrange}{HTML}{F1800A}
\definecolor{lightBlue}{HTML}{0174CD}
\definecolor{greenF}{HTML}{2CBB5C}
\definecolor{cyan}{HTML}{86A6D5}

\tikzstyle{sortnode} = [
  rectangle,
  rounded corners=5pt,
  minimum width=18mm,
  minimum height=6mm,
  very thick,
]

\tikzstyle{functnode} = [
  rectangle,
  minimum width=1cm,
  minimum height=5mm,
]

\tikzstyle{idnode} = [
  sortnode,
  draw=deepblue,
  fill=cyan!50,
  text=deepblue,  
]

\tikzstyle{valnode} = [
  sortnode,
  densely dashed,
  draw=violet,
  fill=magenta!50,
  text=violet,
]

\tikzstyle{f} = [
  thick,
]

\tikzstyle{fd} = [
  f,
  -angle 90,
]

\tikzstyle{relation}=[rectangle split, rectangle split parts=#1, rectangle split part align=base, draw, anchor=center, align=center, text height=3mm, font=\bfseries, text centered]

\tikzstyle{sortnode} = [
  rectangle,
  rounded corners=5pt,
  minimum width=18mm,
  minimum height=6mm,
  very thick,
]

\tikzstyle{functnode} = [
  rectangle,
  minimum width=1cm,
  minimum height=5mm,
]

\tikzstyle{idnode} = [
  sortnode,
  draw=deepblue,
  fill=cyan!50,
  text=deepblue,  
]

\tikzstyle{artnode} = [
  sortnode,
  draw=deepblue,
  fill=yellow!50,
  text=deepblue,  
]

\tikzstyle{valnode} = [
  sortnode,
  densely dashed,
  draw=violet,
  fill=magenta!50,
  text=violet,
]

\tikzstyle{f} = [
  thick,
]

\tikzstyle{fd} = [
  f,
  -angle 90,
]

\tikzstyle{relation}=[rectangle split, rectangle split parts=#1, rectangle split part align=base, draw, anchor=center, align=center, text height=3mm, font=\bfseries, text centered]

\begin{document}

%\title{Model Completeness, Covers and Superposition}

% Keywords:
% - data-centric business processes
% - data abstraction
% - formal verification
% - process modeling and execution

\maketitle

%\category{D.2.4}{Software Engineering}{Software/Program Verification}

\begin{abstract}
Running verification tasks in database driven systems requires solving quantifier elimination 
problems  of a new kind. These quantifier elimination problems are related to the notion of a cover introduced in ESOP 
2008 by Gulwani and Musuvathi.

In this paper,  we show how covers are strictly related to model completions, a well-known topic in 
%classical 
model theory. We also investigate the computation of covers within the Superposition Calculus, by adopting a constrained version of the calculus, equipped with appropriate settings and reduction strategies.

In addition, we show that cover computations are computationally tractable for the fragment of the language used in applications to database driven verification.
This observation 
%silvio 25/6/18 could be 
is confirmed by analyzing the preliminary results obtained using the \mcmt tool 
on the verification of data-aware process benchmarks. These benchmarks can be found in the last version of 
the tool distribution.  
\end{abstract}

\section{Introduction}

Declarative approaches to infinite state model checking~\cite{RV} need to manipulate logical formulae in order to represent sets of reachable states. To prevent divergence, various abstraction strategies have been adopted, ranging from interpolation-based~\cite{McM} to sophisticated search via counterexample elimination~\cite{pdr}. Precise computations of the set of reachable states  require some form of quantifier elimination and hence are subject to two problems, namely  that quantifier elimination might not be available at all and that, when available, it is computationally very expensive.

To cope with the first problem, \cite{GM} introduced the notion of a \emph{cover} and proved that covers exist for equality with uninterpreted symbols (EUF) and its combination with linear arithmetic; also, 
%silvio 15/2
it was shown that  covers can 
be used instead of quantifier elimination and yield a precise computation of reachable states. Concerning the second problem, 
%already 
in~\cite{GM} it was observed (as a side remark) that  computing the cover of a conjunction of literals 
becomes tractable
%substantially lowers down to the complexity of congruence closure 
when only free unary function symbols occur in the signature. It can be shown 
(see Section~\ref{app:complexity} below) 
%(see~\cite{new})
that the same observation applies when also free relational symbols occur. 

In~\cite{CGGMR18,CGGMR19} we propose a new formalism for representing \emph{read-only database schemata} towards the verification of integrated models of processes and data \cite{CaDM13}, in particular so-called \emph{artifact systems}~\cite{Vian09,DHPV09,verifas,boj}; this formalism (briefly recalled in Section~\ref{sec:readonly} below) uses precisely signatures comprising unary function symbols and free $n$-ary relations. In~\cite{CGGMR18,CGGMR19} we apply model completeness techniques for verifying  transition systems based on read-only databases,  in a framework where such systems employ  both individual and higher order variables.

 In this paper we show (see Section~\ref{sec:covers} below) that covers are strictly related to \emph{model completions} and to \emph{uniform interpolation}~\cite{pitts}, thus building a bridge between different research areas. %AGfeb
 In particular, we prove that computing covers for a theory is \emph{equivalent} to eliminating quantifiers in its model completion.
Model completeness has other well-known applications in computer science. It has been applied:
\begin{inparaenum}[\it (i)]
  \item to reveal interesting connections between temporal logic and monadic second order logic~\cite{GvG16,GvG17};
  \item in automated reasoning to design complete algorithms for constraint satisfiability in combined theories over non disjoint signatures~\cite{Ghil05,BGT06,GNZ08,NRR10,NRR09,NRRtac09} and theory extensions~\cite{SS16,Stokker18}; 
  \item to obtain combined interpolation for modal logics and software verification theories~\cite{GG17,GG18}.
\end{inparaenum}

In the last part of the paper (Section~\ref{sec:SC} below), we prove  that covers for (EUF) can be computed through a constrained version of the \emph{Superposition Calculus}~\cite{NR} equipped with appropriate settings and reduction strategies; the related completeness proof requires a careful analysis of the constrained literals generated during the saturation process. Complexity bounds for the fragment used in database driven verification are investigated in Section~\ref{app:complexity}; in Subsection~\ref{subsec:mcmt} we give some details about our first implementation in our tool \textsc{mcmt}.

This paper is the extended version of~\cite{CADE19}.
%Not all proofs could be included here: for the missing ones, we refer to the online available extended version~\cite{new} (the proofs of our results from Section~\ref{sec:SC} are however reported in full detail).

\section{Preliminaries}
\label{sec:prelim}

We adopt the usual first-order syntactic notions of signature, term,
atom, (ground) formula,  and so on; our signatures are multi-sorted and include equality for every sort.
%
%We consider first-order multi-sorted signatures and formulae built up from them. 
This implies that variables are sorted as well. 
%silvio 25/6/18
For simplicity, most basic definitions in this Section will be supplied for single-sorted languages only (the adaptation to  multi-sorted languages is straightforward).
We compactly represent a tuple $\tup{x_1,\ldots,x_n}$ of variables as $\ux$. The notation $t(\ux), \phi(\ux)$ means that the term $t$, the formula $\phi$ has free variables included in the tuple $\ux$.

We assume that a function arity can be deduced from the context. 
Whenever we build terms and formulae, we always assume that they are well-typed, in the sense that the sorts of variables, constants, and function sources/targets match.%}
A formula is said to be \emph{universal} (resp., \emph{existential}) if it has the form $\forall \ux (\phi(\ux))$ (resp., $\exists \ux (\phi(\ux))$), where $\phi$ is a quantifier-free formula. Formulae with no free variables are called \emph{sentences}. 

From the semantic side, we use the standard notion of a $\Sigma$-structure $\cM$ and of truth of a formula in a $\Sigma$-structure under a free variables assignment. 

A \emph{$\Sigma$-theory} $T$ is a set of $\Sigma$-sentences; a \emph{model}  of $T$ is a $\Sigma$-structure $\cM$ where all sentences in $T$ are true.
	 We use the standard notation $T\models \phi$ to say that $\phi$ is true in all models of $T$ for every assignment to the variables occurring free in $\phi$. We say that $\phi$ is \emph{$T$-satisfiable} iff there is a model $\cM$ of $T$ and an assignment to the variables occurring free in $\phi$ making $\phi$ true in $\cM$.
	 
	 We give now the definitions of constraint satisfiability problem and quantifier elimination for a theory $T$.
	 
	 A $\Sigma$-formula $\phi$ is a $\Sigma$-\emph{constraint} (or just a constraint) iff it is a conjunction of literals.
	The \emph{constraint satisfiability problem} for $T$ is the following: we are given a constraint (equivalently, a quantifier-free formula) $\phi(\ux)$
	%an existential formula\footnote{For the purposes of this definition, we may equivalently take
	%$\phi$ to be quantifier-free.
	%}  
	%$\exists \uy \,\phi(\ux,\uy)$ 
	and we are asked whether there exist  a model $\cM$ of $T$ and an assignment $\cI$ to the free variables $\ux$ such that $\cM, \cI \models \phi(\ux)$.

	 A theory $T$ has \emph{quantifier elimination} iff for every formula $\phi(\ux)$ in the signature of $T$ there is a quantifier-free formula 
	$\phi'(\ux)$ such that $T\models \phi(\ux)\leftrightarrow \phi'(\ux)$. It is well-known (and easily seen) that quantifier elimination holds in case we can eliminate quantifiers from \emph{primitive} formulae, i.e. from formulae of the kind $\exists \uy \,\phi(\ux, \uy)$, where $\phi$ is a conjunction of literals (i.e. of atomic formulae and their negations). Since we are interested in effective computability, we assume that when we talk about quantifier elimination, an effective procedure for eliminating quantifiers is given.

 We recall also some basic definitions and notions from logic and model theory. We focus on the definitions of diagram, embedding, substructure and amalgamation.
 
 \subsection{Substructures and embeddings}
 
 Let $\Sigma$ be a first-order signature. The signature
 obtained from $\Sigma$ by adding to it a set $\ua$ of new constants
 (i.e., $0$-ary function symbols) is denoted by $\Sigma^{\ua}$.
Analogously, given a $\Sigma$-structure $\cM$, the signature $\Sigma$ can be expanded to a new signature $\Sigma^{|\cM|}:=\Sigma\cup \{\bar{a}\ |\ a\in |\cM| \}$ by adding a set of new constants $\bar{a}$ (the \textit{name} for $a$), one for each element $a$ in $\cM$, with the convention that two distinct elements are denoted by different "name" constants. $\cM$ can be expanded to a $\Sigma^{|\cM|}$-structure $\overline{\cM}:=(\cM, a)_{a\in \vert\cM\vert}$ just interpreting the additional constants over the corresponding elements. From now on, when the meaning is clear from the context, we will freely use the notation  $\cM$ and  $\overline{\cM}$ interchangeably: in particular, given a $\Sigma$-structure
$\cM$ 
and a $\Sigma$-formula $\phi(\ux)$ with free variables that are all in $\ux$, we will write, by abuse of notation, 
$\cM\models \phi(\ua)$ instead of $\overline{\cM}\models \phi(\bar{\ua})$.

%$\cM\models \phi(\ua)$ 

A {\it $\Sigma$-homomorphism} (or, simply, a
homomorphism) between two $\Sigma$-structu\-res $\cM$ and
$\cN$ is any mapping $\mu: \vert \cM \vert \lra \vert \cN\vert $ among the
support sets $\vert \cM \vert $ of $\cM$ and $\vert \cN \vert$  of $\cN$ satisfying the condition
\begin{equation}\label{eq:emb}
\cM \models \varphi \quad \Rightarrow \quad \cN \models \varphi
\end{equation}
for all $\Sigma^{\vert \cM\vert}$-atoms $\varphi$ (here $\cM$ is regarded as a
$\Sigma^{\vert \cM\vert}$-structure, by interpreting each additional constant $a\in
\vert \cM\vert $ into itself and $\cN$ is regarded as a $\Sigma^{\vert \cM\vert}$-structure by
interpreting each additional constant $a\in \vert \cM\vert $ into $\mu(a)$). 
In case condition \eqref{eq:emb} holds for all $\Sigma^{|\cM|}$-literals,
the homomorphism $\mu$ is said to be an {\it
	embedding} and if it holds for all 
first order
formulae, the embedding $\mu$ is said to be {\it
	elementary}. 
Notice the following facts: 
\begin{description}
	\item[{\rm (a)}] since we have equality in the
	signature, an embedding is an injective function; 
	\item[{\rm (b)}] an embedding $\mu:\cM\longrightarrow \cN$ must be an algebraic homomorphism, 
	that is for every $n$-ary function symbol $f$ and for every $m_1,..., m_n$ in $|\cM|$, we must have
	$f^{\cN}(\mu(m_1), ... , \mu(m_n)) = \mu(f^{\cM}(m_1, ... , m_n))$; 
	\item[{\rm (c)}] for an $n$-ary predicate symbol $P$ we must have $(m_1, ... , m_n) \in P^{\cM}$ iff  $(\mu(m_1), ... , 
	\mu(m_n)) \in P^{\cN}$.
\end{description}
It is easily seen that an embedding  $\mu:\cM\longrightarrow \cN$ can be equivalently
defined as a map  $\mu:\vert\cM\vert\longrightarrow\vert \cN\vert$ satisfying the conditions (a)-(b)-(c) above.
If  $\mu: \cM \lra \cN$ is an embedding which is just the
identity inclusion $\vert \cM\vert\subseteq\vert \cN\vert$, we say that $\cM$ is a {\it
	substructure} of $\cN$ or that $\cN$ is an {\it extension} of
$\cM$. A $\Sigma$-structure $\cM$ is said to be \emph{generated by}
a set $X$ included in its support $\vert \cM\vert$ iff there are no proper substructures of $\cM$ including $X$.

The notion of substructure can be equivalently defined as follows: given a $\Sigma$-structure $\cN$ and a $\Sigma$-structure $\cM$ such that $|\cM| \subseteq |\cN|$, we say that $\cM$ is a $\Sigma$-\textit{substructure} of $\cN$ if:

\begin{itemize}
	\item for every function symbol $f$ inf $\Sigma$, the interpretation of $f$ in $\cM$ (denoted using $f^{\cM}$) is the restriction of the interpretation of $f$ in $\cN$ to $|\cM|$ (i.e. $f^{\cM}(m)=f^{\cN}(m)$ for every $m$ in $|\cM|$); this fact implies that a substructure $\cM$ must be a subset of $\cN$ which is closed under the application of $f^{\cN}$.
	\item  for every relation symbol $P$ in $\Sigma$ and every tuple $(m_1,...,m_n)\in |\cM|^{n}$, $(m_1,...,m_n)\in P^{\cM}$ iff $(m_1,...,m_n)\in P^{\cN}$, which means that the relation $P^{\cM}$ is the restriction of $P^{\cN}$ to the support of $\cM$.
\end{itemize} 
 
We recall that a substructure \textit{preserves} and \textit{reflects} validity of ground formulae, in the following sense:
given a $\Sigma$-substructure $\cM_1$ of a $\Sigma$-structure $\cM_2$, a ground $\Sigma^{\vert\cM_1\vert}$-sentence $\theta$ is true in $\cM_1$ iff $\theta$ is true in $\cM_2$.

\subsection{Robinson Diagrams and Amalgamation}

Let $\cM$ be a $\Sigma$-structure. The \textit{diagram} of $\cM$, denoted by $\Delta_{\Sigma}(\cM)$, is defined as the set of ground $\Sigma^{|\cM|}$-literals (i.e. atomic formulae and negations of atomic formulae) that are true in $\cM$.  
%For the sake of simplicity, once again by abuse of notation, we will freely say that  $\Delta_{\Sigma}(\cA)$ is the set of $\Sigma^{|\cA|}$-literals which are true in $\cA$.

An easy but nevertheless important basic result, called 
\emph{Robinson Diagram Lemma}~\cite{CK},
says that, given any $\Sigma$-structure $\cN$,  the embeddings $\mu: \cM \longrightarrow \cN$ are in bijective correspondence with
expansions of $\cN$ 
to   $\Sigma^{\vert \cM\vert}$-structures which are models of 
$\Delta_{\Sigma}(\cM)$. The expansions and the embeddings are related in the obvious way: $\bar a$ is interpreted as $\mu(a)$.

Amalgamation is a classical algebraic concept. We give the formal definition of this notion.
\begin{definition}[Amalgamation]
	\label{def:amalgamation}
	A theory $T$ has the \emph{amalgamation property} if for every couple of embeddings $\mu_1:\cM_0\longrightarrow\cM_1$, $\mu_2:\cM_0\longrightarrow\cM_2$ among models of $T$, there exists a model $\cM$ of
	$T$ endowed with embeddings $\nu_1:\cM_1 \longrightarrow \cM$ and
	$\nu_2:\cM_2 \longrightarrow \cM$ such that $\nu_1\circ\mu_1=\nu_2\circ\mu_2$
	\begin{center}
		\begin{tikzcd}
			& \cM \arrow[rd, leftarrow, "\nu_2"] & \\
			\cM_{1} \arrow[ru, "\nu_1"] \arrow[rd, leftarrow, "\mu_1"]	& & \cM_{2} &\\
			& \cM_{0} \arrow[ur, "\mu_2"] &
		\end{tikzcd}
	\end{center}
\end{definition}
The triple   $(\cM, \mu_1,\mu_2)$ (or, by abuse,   $\cM$  itself) is said to be a $T$-amalgama of  $\cM_1,\cM_2$ over $\cM_0$
%Introduce the notation $\cM\models \phi(\ua)$ ($\ua$ are parameters from the support of $\cM$, instead of assignments).

\section{Covers, Uniform Interpolation and Model Completions}\label{sec:covers}

We report the notion of \emph{cover} taken from~\cite{GM}. 
Fix a  theory $T$ and an existential formula $\exists \ue\, \phi(\ue, \uy)$; call a \emph{residue} of $\exists \ue\, \phi(\ue, \uy)$ any quantifier-free formula belonging to the set of quantifier-free formulae  $Res(\exists \ue\, \phi)=\{\theta(\uy, \uz)\mid T \models \phi(\ue, \uy) \to \theta(\uy, \uz)\}$. 
A quantifier-free formula $\psi(\uy)$ is said to be a \emph{$T$-cover} (or, simply, a \emph{cover}) of $\exists \ue\, \phi(\ue,\uy)$ iff  $\psi(\uy)\in Res(\exists \ue\, \phi)$ and $\psi(\uy)$ implies (modulo $T$) all the other formulae in $Res(\exists \ue\, \phi)$.
The following Lemma (to be widely used throughout the paper) supplies a semantic counterpart to the notion of a cover:

\begin{lemma}\label{lem:cover} A formula $\psi(\uy)$ is a $T$-cover of $\exists \ue\, \phi(\ue, \uy)$ iff 
it satisfies the following two conditions:
\begin{inparaenum}[(i)]
\item $T\models  \forall \uy\,( \exists \ue\,\phi(\ue, \uy) \to \psi(\uy))$;
\item for every model $\cM$ of $T$, for every tuple of  elements $\ua$ from the support of $\cM$ such that $\cM\models \psi(\ua)$ it is possible to find
  another model $\cN$ of $T$ such that $\cM$ embeds into $\cN$ and $\cN\models \exists \ue \,\phi(\ue, \ua)$.
\end{inparaenum}
\end{lemma}

\begin{proof}
\textit{Suppose that $\psi(\uy)$ satisfies conditions (i) and (ii) above}. 
Condition (i) says that $\psi(\uy)\in Res(\exists \ue\, \phi)$, so $\psi$ is a residue.
 In order to show that $\psi$ is also a cover, we have to prove that $T\models \forall \uy,\uz (\psi(\uy)\to \theta(\uy,\uz))$, for every $\theta(\uy,\uz)$ that is a residue for $\exists \ue\, \phi(\ue, \uy)$.
 Given a model $\cM$ of $T$, take a pair of tuples $\ua, \ub$ of elements from $|\cM|$ and suppose that $\cM\models \psi(\ua)$. By condition (ii), there is a model $\cN$ of $T$ such that $\cM$ embeds into $\cN$ and $\cN\models \exists \ue \phi(\ue, \ua)$. Using the definition of $Res(\exists \ue\, \phi)$, we have $\cN\models \theta(\ua,\ub)$, since $\theta(\uy,\uz)\in Res(\exists \ux\, \phi)$. Since $\cM$ is a substructure of $\cN$ and $\theta$ is quantifier-free, $\cM\models \theta(\ua,\ub)$ as well, as required.

\textit{Suppose that $\psi(\uy)$ is a cover}. The definition of residue implies condition (i). To show condition (ii) we have to prove that, given a model $\cM$ of $T$, for every tuple $\ua$ of elements from $|\cM|$, if $\cM\models \psi(\ua)$, then there exists a model $\cN$ of $T$ such that $\cM$ embeds into $\cN$ and $\cN\models \exists \ux \phi(\ux, \ua)$. By reduction to absurdity, suppose that this is not the case: this is equivalent (by using Robinson Diagram Lemma) to the fact that $\Delta(\cM)\cup \{ \phi(\ue, \ua) \}$ is a $T$-inconsistent $\Sigma^{|\cM|\cup\{\ue\}}$-theory. By compactness, there is a finite number of literals $\ell_1(\ua,\ub),...,\ell_m(\ua,\ub)$ (for some tuple $\ub$ of elements from $|\cM|$) such that $\cM\models \ell_i$ (for all $i=1,\dots,m$) and  $T\models\phi(\ue, \ua)\to \neg (\ell_1(\ua,\ub)\land\cdots\land \ell_m(\ua, \ub))$, which means that 
        $T\models \phi(\ue, \uy)\to (\neg \ell_1(\uy,\uz)\lor \cdots \lor \neg \ell_m(\uy,\uz))$, i.e. that 
        $T\models \exists\ue\,\phi(\ue, \uy)\to (\neg \ell_1(\uy,\uz)\lor \dots \lor \neg \ell_m(\uy,\uz))$. By definition of residue, clearly $(\neg \ell_1(\uy,\uz)\lor \dots \lor \neg \ell_m(\uy,\uz)) \in Res(\exists \ux\, \phi)$; then, since $\psi(\uy)$ is a cover, $T\models \psi(\uy)\to (\neg \ell_1(\uy,\uz)\lor \dots \lor \neg \ell_m(\uy,\uz))$, which implies that $\cM \models \neg \ell_j(\ua,\ub)$ for some $j=1,\dots,m$, which is a contradiction. Thus, $\psi(\uy)$ satisfies conditions (ii) too. 
\end{proof}

We say that a theory $T$ has \emph{uniform quantifier-free interpolation} iff every existential formula $\exists \ue\, \phi(\ue,\uy)$ (equivalently, every primitive formula $\exists \ue\, \phi(\ue,\uy)$) has a $T$-cover.

It is clear that if $T$ has uniform quantifier-free interpolation, then it has ordinary quantifier-free interpolation~\cite{BGR14}, in the sense that if we have $T\models \phi(\ue, \uy)\to \phi'(\uy, \uz)$ (for quantifier-free formulae $\phi, \phi'$), then there is a quantifier-free formula $\theta(\uy)$ such that  $T\models \phi(\ue, \uy)\to \theta(\uy)$ and $T\models \theta(\uy)\to \phi'(\uy, \uz)$. In fact, if $T$ has uniform quantifier-free interpolation, then the interpolant $\theta$ is independent on $\phi'$ (the  same $\theta(\uy)$ can be used as interpolant for all entailments $T\models \phi(\ue, \uy)\to \phi'(\uy, \uz)$, varying $\phi'$).

 We say that a \emph{universal} theory $T$ has a \emph{model completion} iff there is a stronger theory $T^*\supseteq T$ (still within the same signature $\Sigma$ of $T$) such that (i) every $\Sigma$-constraint that is satisfiable
	in a model of $T$ is satisfiable in a model of $T^*$; (ii) $T^*$ eliminates quantifiers.
	Other equivalent definitions are possible~\cite{CK}: for instance, (i) is equivalent to the fact that $T$ and $T^*$ prove the same quantifier-free formulae or again to the fact that every model of $T$ can be embedded into a model of $T^*$.
	We recall that the model completion, if it exists, is unique and that its existence implies the amalgamation property for $T$~\cite{CK}. The relationship between uniform interpolation in a propositional logic and model completion 
	%silvio 15/2
	of the the equational theory 
	of the variety algebraizing it was extensively studied in~\cite{GZ}. In  the context of first order theories, we %have %AG20feb
	prove an even more direct connection:

\begin{theorem}\label{prop:qe} Suppose that $T$ is a universal theory. Then $T$ has a model completion $T^*$ iff $T$ has uniform quantifier-free interpolation. If this happens, 
  $T^*$ is axiomatized by the infinitely many sentences
  %~\footnote{Notice that our $T$ is assumed to be universal according to Definition~\ref{def:db}, whereas $T^*$ turns out to be universal-existential. }
  \begin{equation}\label{eq:qe}
  \forall \uy \,(\psi(\uy) \to \exists \ue\, \phi(\ue, \uy))
  \end{equation}
  where $\exists \ue\, \phi(\ue, \uy)$ is a  primitive formula and $\psi$ is a cover of it.
\end{theorem}

\begin{proof}
 Suppose first that there is a model completion $T^*$ of $T$ and let $\exists \ue\,\phi(\ue,\uy)$ be a primitive formula. Since $T^*$ eliminates quantifiers, we have 
 $T^*\models \exists \ue\,\phi(\ue,\uy) \leftrightarrow \psi(\uy)$ for some quantifier-free formula $\psi(\uy)$. Since $T$ and $T^*$ prove the same quantifier-free formulae, we have that $\psi(\uy)\in Res(\exists \ue\,\phi)$. If $\theta(\uy,\uz)\in Res(\exists\ue \,\phi)$, then we have 
 $T\models \phi(\ue,\uy)\to \theta(\uy, \uz)$; the same entailment holds in $T^*$ too, where we have $T^*\models \psi(\uy)\to \theta(\uy,\uz)$. Since $\psi(\uy)\to \theta(\uy,\uz)$ is quantifier-free, we have also $T\models \psi(\uy)\to \theta(\uy,\uz)$, showing that $\psi$ is a cover of $\exists \ue\,\phi(\ue,\uy)$.
 Thus $T$ has uniform interpolation, because we found a cover for every primitive formula.

Suppose vice versa that $T$ has uniform interpolation. Let $T^*$ be the theory axiomatized by all the formulae~\eqref{eq:qe} above.
 %By a chain argument. 
 From {\rm (i)} of Lemma~\ref{lem:cover} and (\ref{eq:qe}) above, we clearly get that $T^{\star}$ admits quantifier elimination: in fact, in order to prove that a theory enjoys quantifier elimination, it is sufficient to eliminate quantifiers from \textit{primitive} formulae (then the quantifier elimination for all formulae can be easily shown by an induction over their complexity). This is exactly what is guaranteed by {\rm (i)} of Lemma~\ref{lem:cover} and (\ref{eq:qe}).
 
 Let $\cM$ be a model of $T$. We show (by using a chain argument) that there exists a model $\cM^{\prime}$ of $T^{\star}$ such that $\cM$ embeds into $\cM^{\prime}$. For every primitive formula $\exists \ue \,\phi(\ue,\uy)$, consider the set $\{(\ua, \exists \ue \,\phi(\ue, \ua))\}$ such that $\cM\models \psi(\ua)$ 
 %silvio 15/2
 (where $\psi$ is a cover of $\phi$).
 %related to $\phi$ as in {\rm (i)} of Lemma~\ref{lem:cover}). 
 By Zermelo's Theorem, the set  $\{(\ua, \exists \ue\,\phi(\ue, \ua))\}$ can be well-ordered: let $\{(\ua_i,\exists \ue\,\phi_i(\ue, \ua_i))\}_{i\in I}$ be such a well-ordered set (where $I$ is an ordinal). By transfinite induction on this well-order, we define $\cM_{0}:=\cM$ and, for each $i\in I$, 
 %silvio 15/2 bisogna tener conto anche degli ordinali limite....
 $\cM_{i}$ as the extension of $\bigcup_{j<i}\cM_j$ such that $\cM_{i}\models\exists \ue\, \phi_i(\ue, \uy)$, which exists for {\rm (ii)} of Lemma~\ref{lem:cover}  since $\bigcup_{j<i}\cM_j\models \psi_i(\ua)$ (remember that validity of ground formulae is preserved passing through substructures and superstructures, and $\cM_0\models \psi_i(\ua)$).

 Now we take the chain union $\cM^{1}:=\bigcup_{i\in I} \cM_{i}$: since $T$ is universal, $\cM^{1}$ is again a model of $T$, and it is possible to construct an analogous chain $\cM^{2}$ as done above, starting from $\cM^{1}$ instead of $\cM$. Clearly, we get $\cM_{0}:=\cM\subseteq \cM^{1} \subseteq \cM^2$ by construction. At this point, we iterate the same argument countably many times, so as to define a new chain of models of $T$: $$\cM_{0}:=\cM\subseteq \cM^{1}\subseteq...\subseteq \cM^n\subseteq...$$
 
 Defining $\cM^{\prime}:=\bigcup_n \cM^{n}$, we trivially get that $\cM^{\prime}$ is a model of $T$ such that $\cM\subseteq \cM^{\prime}$ and satisfies all the sentences of type (\ref{eq:qe}). The last fact can be shown using the following finiteness argument.
 
 Fix $\phi, \psi$ as in~\eqref{eq:qe}.
 For every tuple $\ua^{\prime}\in \cM^{\prime}$ such that $\cM^{\prime}\models \psi(\ua^{\prime})$, by definition of $\cM^{\prime}$ there exists a natural number $k$ such that $\ua^{\prime}\in \cM^{k}$: since $\psi(\ua^{\prime})$ is a ground formula, we get that also $\cM^{k}\models \psi(\ua^{\prime})$. Therefore, we consider the step $k$ of the countable chain: there, we have that the pair $(\ua^{\prime},\psi(\ua^{\prime}))$ appears in the enumeration given by the well-ordered set 
 %silvio 15/2
 of pairs
 $\{(\ua_i,\exists \ue\, \phi_i(\ue,\ua_i))\}_{i\in I}$ (for some ordinal $I$) such that $\cM^{k}\models \psi_i(\ua)$. Hence, by construction and since $\psi(\ua^{\prime})$ is a ground formula, we have that there exists a $j\in I$ such that
 %silvio 1572
 %$\cM^{k}_{j}\models \psi(\ua^{\prime})$ and $\cM^{k}_{j+1}\models \exists \ue\phi(\ue,\ua^{\prime})$.
 $\cM^{k}_{j}\models \exists \ue\,\phi(\ue,\ua^{\prime})$.
 In conclusion, since the existential formulae are preserved passing to extensions, we obtain $\cM^{\prime}\models\exists \ue\,\phi(\ue,\ua^{\prime})$, as wanted.
% $\hfill \dashv$
\end{proof}

\section{Model-Checking Applications}

In this section we supply old and new motivations for investigating covers and model completions in view of model-checking applications.
We first report the considerations from~\cite{GM,CGGMR18,CGGMR19} on symbolic model-checking %AG:20feb
via model completions (or, equivalently, via covers) in the basic case where system variables are represented as individual variables (for more advanced applications where system variables are both individual and higher order variables, see~\cite{CGGMR18,CGGMR19}). Similar ideas (`use quantifier elimination  in the
model completion even if $T$ does not allow quantifier elimination') were  used in~\cite{SS16} for
interpolation and symbol elimination.

\begin{definition}\label{def:simpleartifact}
  A (quantifier-free) \emph{transition system}  is a tuple
  \begin{align*}
    \cS ~=~\tup{\Sigma, T, \ux, \iota(\ux), \tau(\ux, \ux')}
  \end{align*}
  where:
  \begin{inparaenum}[\itshape (i)]
  \item $\Sigma$ is a signature and $T$ is a $\Sigma$-theory;
  \item $\ux=x_1, \dots, x_n$ are individual variables;
  \item $\iota(\ux)$ is a quantifier-free formula;
  \item $\tau(\ux, \ux')$ is a quantifier-free formula (here the $\ux'$ are renamed copies of the $\ux$).
  \end{inparaenum}
\end{definition}

\medskip
\noindent
A \emph{safety} formula for a transition system  $\cS$ is a further quantifier-free formula $\upsilon(\ux)$
describing undesired states of $\cS$. We say that $\cS$ is \emph{safe with
respect to} $\upsilon$ if the system has no finite run leading from
$\iota$ to $\upsilon$, i.e.   (formally) if there are no model $\cM$ of $T$ and
 no $k\geq 0$
% , and no assignment in $\cM$ to the variables
%$\ux^0, \dots, \ux^k$ 
such that the formula
\begin{equation}\label{eq:smc}
  \iota(\ux^0) \land \tau(\ux^0, \ux^1) \land \cdots \land\tau(\ux^{k-1},
  \ux^k)\land \upsilon(\ux^k)
\end{equation}
is satisfiable in $\cM$ (here $\ux^i$'s are renamed copies of $\ux$). The \emph{safety
 problem} for $\cS$ is the following: \emph{given 
 %a safety formula 
 $\upsilon$,
 decide whether $\cS$ is safe with respect to $\upsilon$}.

 Suppose now that the theory $T$ mentioned in Definition~\ref{def:simpleartifact}(i) is universal, has decidable constraint satisfability problem and admits a model completion $T^*$. 
Algorithm~\ref{alg1} describes the \emph{backward reachability algorithm}  for handling the safety problem for $\cS$ (the dual algorithm working via forward search is described in equivalent terms in~\cite{GM}).  An integral
part of the algorithm is to compute preimages.
For that purpose, for any $\phi_1(\ux,\ux')$ and $\phi_2(\ux)$, we define
% (where $\uz'$ are renamed copies of $\uz$),
$\mathit{Pre}(\phi_1,\phi_2)$ to be the formula
$\exists \ux'(\phi_1(\ux, \ux')\land \phi_2(\ux'))$.
The \emph{preimage} of the set of states described by a state formula
$\phi(\ux)$ is the set of states described by
$\mathit{Pre}(\tau,\phi)$.
%\footnote{Notice that, when $\tau=\bigvee\hat\tau$,
% then $\mathit{Pre}(\tau,\phi)=\bigvee\mathit{Pre}(\hat\tau,\phi)$.} 
 \begin{wrapfigure}[12]{r}{0.45\textwidth}
 \vspace{-22pt}
 \small
\begin{algorithm}[H]
\SetKwProg{Fn}{Function}{}{end}
\Fn{$\mathsf{BReach}(\upsilon)$}{
\setcounter{AlgoLine}{0}
\ShowLn$\phi\longleftarrow \upsilon$;  $B\longleftarrow \bot$\;
\ShowLn\While{$\phi\land \neg B$ is $T$-satisfiable}{
\ShowLn\If{$\iota\land \phi$ is $T$-satisfiable.}
{\textbf{return}  $\mathsf{unsafe}$}
\setcounter{AlgoLine}{3}
\ShowLn$B\longleftarrow \phi\vee B$\;
\ShowLn$\phi\longleftarrow \mathit{Pre}(\tau, \phi)$\;
\ShowLn$\phi\longleftarrow \mathsf{QE}(T^*,\phi)$\;
}
\textbf{return} $(\mathsf{safe}, B)$;}{
\caption{Backward reachability algorithm}\label{alg1}
}
\end{algorithm}
 \end{wrapfigure}
 The
subprocedure $\mathsf{QE}(T^*,\phi)$ in Line~6 applies the quantifier
elimination algorithm of $T^*$ to the existential formula $\phi$.
Algorithm~\ref{alg1} computes iterated preimages of $\upsilon$ and applies to
them quantifier elimination, until a fixpoint is reached or until a set
intersecting the initial states (i.e., satisfying $\iota$) is found.
\textit{Inclusion} (Line~2) and \textit{disjointness} (Line~3)
tests produce proof obligations that can be 
discharged thanks to the fact that $T$ has decidable constraint satisfiability problem.

The proof of Proposition~\ref{thm:basic} consists just in the %AGfeb simple 
observation that, thanks to quantifier elimination in $T^\star$,~\eqref{eq:smc} is a quantifier-free formula 
and that %satisfiability of
a quantifier-free formula is satisfiable in a model of $T$ %is the same as its satisfiability 
 iff so is it in a model of $T^*$:

\begin{proposition}\label{thm:basic} Suppose that the universal $\Sigma$-theory $T$ has decidable constraint satisfiability problem and admits a model completion $T^*$.
  For every transition system
  $\cS=\tup{\Sigma,T,\ux,\iota,\tau}$, the  backward search algorithm is effective and partially correct for solving safety
    problems for $\cS$.\footnote{\emph {Partial correctness} means that, when
     the algorithm terminates, it gives a correct answer. \emph{Effectiveness}
     means that all subprocedures in the algorithm can be effectively
     executed.}
\end{proposition}

Despite its simplicity, Proposition~\ref{thm:basic} is a crucial fact. Notice that it implies decidability of the safety problems in some interesting cases: this happens, for instance, when in $T$ %AGfeb 
there are only finitely many quantifier-free formulae in which $\ux$ occur, as in case $T$ has a purely relational signature or, more generally, $T$ is \emph{locally finite}\footnote{We say that 
$T$ is locally finite iff for every finite tuple of variables $\ux$ there are only finitely many non $T$-equivalent atoms 
$A(\ux)$ involving only the variables $\ux$.}. Since a theory is universal iff it is closed under substructures~\cite{CK} and since a universal 
locally finite theory has  a model completion iff it has the amalgamation property~\cite{wheeler,LIP}, %AGfeb: ho aggiunto il riferimento al teorema 5 per essere precisi e evitare vecchie critiche
it follows that Proposition~\ref{thm:basic} can be used to cover the decidability result stated in Theorem 5 of~\cite{boj} (once restricted to transition systems over a first-order definable class of $\Sigma$-structures).

\begin{subsection}{Database Schemata}
\label{sec:readonly}
In this subsection, we provide a new application for the above explained model-checking techniques~\cite{CGGMR18,CGGMR19}. The application relates to the verification of integrated models of business processes and data \cite{CaDM13}, referred to as artifact systems \cite{Vian09}, where the behavior of the process is influenced by data stored in a relational database (DB) with constraints. The data contained therein are read-only: they can be queried by the process and stored in a working memory, which in the context of this paper is constituted by a set of system variables. In this context, safety amounts to checking whether the system never reaches an undesired property, irrespectively of what is contained in the read-only DB.  

We define next the two key notions of (read-only) DB schema and instance, by relying on an algebraic, functional characterization. 	

\begin{definition}\label{def:db}
A \emph{DB schema} is a pair $\tup{\Sigma,T}$, where:
  \begin{inparaenum}[\it (i)]
    \item $\Sigma$ is a \emph{DB signature}, that is, a finite multi-sorted signature whose 
    function symbols are all unary;
    %only symbols are equality, unary functions, and constants;
    \item $T$ is a \emph{DB theory}, that is, a set of universal $\Sigma$-sentences. 
  \end{inparaenum}
\end{definition}
We now focus on extensional data conforming to a given DB schema.	
\begin{definition}	
\label{def:instance}
	A \emph{DB instance} of DB schema $\tup{\Sigma,T}$ is a $\Sigma$-structure $\cM$ such that
	%\begin{inparaenum}[\it (i)]
	%\item 
	$\cM$ is a model of $T$.\footnote{One may restrict to models interpreting sorts as \emph{finite} sets, as customary in database theory.
	Since the theories we are dealing with usually have finite model property for constraint satisfiability, assuming such restriction turns out 
	to be irrelevant, as far as safety problems are concerned (see~\cite{CGGMR18,CGGMR19} for an accurate discussion).
	}
	%, and
    %\item every id sort of $\Sigma$ is interpreted by $\cM$ on a \emph{finite} set.
	%\end{inparaenum}
\end{definition}
\begin{comment}
As usual, a DB instance has to be distinguished from an arbitrary \emph{model} of $T$, where no finiteness assumption is posed on the interpretation of id sorts. What may appear as not customary in Definition~\ref{def:instance} is the fact that value sorts can be interpreted on infinite sets. This allows us, at once, to reconstruct the classical notion of DB instance as a finite model (since only finitely many values can be pointed from id sorts using functions), at the same time supplying a potentially infinite set of fresh values to be dynamically introduced in the working memory during the evolution of the artifact system. 
We respectively denote by $S^\cM$, $f^\cM$, and $c^\cM$ the interpretation in $\cM$ of the sort $S$ (this is a set),
	of the function symbol $f$ (this is a set-theoretic function), and of the constant $c$ (this is an element of the interpretation of the corresponding sort).
	Obviously, $f^\cM$ and $c^\cM$ must match the sorts declared in $\Sigma$. For instance, if the source and the target of $f$ are, respectively, $S$ and $U$, then the function $f^\cM$ has domain $S^\cM$ and range $U^\cM$.
\end{comment}

One might be surprised by the fact that signatures in our DB schemata contain unary function symbols, beside relational symbols. %AGfeb
As shown in~\cite{CGGMR18,CGGMR19}, the algebraic, functional characterization of DB schema and instance can be actually reinterpreted in the classical, relational model so as to reconstruct the requirements posed in~\cite{verifas}. Definition~\ref{def:db} naturally corresponds to the definition of relational database schema equipped with single-attribute \emph{primary keys} and \emph{foreign keys}. %(plus a reformulation of constraint~\eqref{eq:null}). 
To see this connection, we adopt the \emph{named perspective}, where each relation schema is defined by a signature containing a \emph{relation name} and a set of \emph{typed attribute names}. Let $\tup{\Sigma,T}$ be a DB schema. Each 
sort 
%$S \in \ids{\Sigma}$ 
$S$ from $\Sigma$
corresponds to a dedicated relation $R_S$ with the following attributes:
    \begin{inparaenum}[\itshape (i)]
      \item one identifier attribute $id_S$ with type $S$;
      \item one dedicated attribute $a_f$ with type $S'$ for every function symbol 
      $f$ from $\Sigma$
      %$f \in \functs{\Sigma}$ 
      of the form $f: S \longrightarrow S'$.
    \end{inparaenum}

The fact that $R_S$ is constructed starting from functions in $\Sigma$ naturally induces corresponding functional dependencies within $R_S$, and inclusion dependencies from $R_S$ to other relation schemas. 
%In particular, we obtain the following constraints for $R_S$:
%\begin{compactitem}[$\bullet$]
%\item 
In particular, for each non-id attribute $a_f$ of $R_S$, we get a functional dependency from $id_S$ to $a_f$. Altogether, such dependencies witness that $\mathit{id}_S$ is the \emph{primary key} of $R_S$. In addition,
for each non-id attribute $a_f$ of $R_S$ whose corresponding function symbol $f$ has id sort $S'$ as image, we get an inclusion dependency from $a_f$ to the id attribute $id_{S'}$ of $R_{S'}$. This captures that $a_f$ is a \emph{foreign key} referencing $R_{S'}$.

Given a DB instance $\M$ of $\tup{\Sigma,T}$, its corresponding relational instance $\mathcal{R}[{\cM}]$ is the minimal set satisfying the following property:
for every id sort $S$ from $\Sigma$, let $f_1,\ldots,f_n$ be all functions in $\Sigma$
with domain $S$; then, for every identifier $\constant{o} \in S^\M$, $\mathcal{R}[{\cM}]$ contains a  \emph{labeled fact} of the form $R_S(\lentry{id_S}{\constant{o}^\M},\lentry{a_{f_1}}{f_1^\M(\constant{o})},\ldots,\lentry{a_{f_n}}{f_n^\M(\constant{o})})$.	
%sg 30/4/19 aggiunte le righe seguenti (per prevenire la solita obiezione idiota...)
In addition, $\mathcal{R}[{\cM}]$ contains the tuples from $r^\cM$, for every relational symbol $r$ from $\Sigma$ (these relational symbols represent plain relations, i.e. those not possessing a key).
	
	We close our discussion by focusing on DB theories. Notice that EUF suffices to handle the sophisticated setting of database-driven systems from \cite{CGGMR19} (e.g., key dependencies). The role of a non-empty DB theory is to encode background axioms to express additional constraints. We illustrate a typical background axiom, required to handle the possible presence of \emph{undefined identifiers/values} in the different sorts. This, in turn, is essential to capture 
artifact systems
%\aas 
whose working memory is initially undefined, in the style of \cite{DeLV16,verifas}. To accommodate this, we add to every sort $S$ of $\Sigma$ a constant $\nullv_S$ (written by abuse of notation just $\nullv$ from now on), used to specify an undefined value. Then, for each function symbol $f$ of $\Sigma$, we can impose additional constraints involving $\nullv$, for example by adding the following axioms to the DB theory:
	\begin{equation}\label{eq:null}
	 \forall x~(x = \nullv \leftrightarrow f(x) = \nullv)
	\end{equation}
     This axiom states that the application of $f$ to the undefined value produces an undefined value, and it is the only situation for which $f$ is undefined.
 %AGfeb
  A slightly different approach may handle \emph{many} undefined values for each sort; the reader is referred to~\cite{CGGMR18,CGGMR19} for examples of concrete database instances formalized in our framework. We just point out that in most cases the kind of axioms that we need for our DB theories $T$ are just \emph{one-variable universal axioms} (like Axioms~\ref{eq:null}), so that they fit 
  the hypotheses of 
  Proposition~\ref{prop:mc} below.

	\end{subsection}

%%% Local Variables:
%%% mode: latex
%%% TeX-master: "main"
%%% End:

%\subsection{Model Checking over DB schemata}
%\label{app_sec2}

We are interested in applying the algorithm of Proposition~\ref{thm:basic} to what we call \emph{simple artifact systems}, i.e. transition systems 
$\cS ~=~\tup{\Sigma, T, \ux, \iota(\ux), \tau(\ux, \ux')}$, where $\tup{\Sigma,T}$ is a DB schema in the sense of Definition~\ref{def:db}. To this aim, it is sufficient to
identify a suitable class of DB theories having a model completion and whose constraint satisfiability problem is decidable.
A first result in this sense is  given below. 
We associate to a DB signature $\Sigma$ 
%a characteristic graph 
%$G(\Sigma)$ capturing the dependencies that are induced by functions over sorts. Specifically, 
the edge-labeled graph 
$G(\Sigma)$ whose nodes are the sorts in ${\Sigma}$, and such that $G(\Sigma)$ contains a labeled edge $S \xrightarrow{f} S'$ if and only if $\Sigma$ contains a function symbol 
whose source sort is $S$ and whose target sort is $S'$.
%$f:S\longrightarrow S'$. 
We say that $\Sigma$ is \emph{acyclic} if $G(\Sigma)$ is so. 
%The \emph{leaves} of $\Sigma$ are the nodes of $G(\Sigma)$ without outgoing edges.

\begin{comment}
 The  discussion at the end of Subsection~\ref{sec:readonly} shows that often the DB theories we are interested in are rather simple: they include finitely many 
one-variable universal axioms like~\eqref{eq:null}, \eqref{eq:null1}. If, in addition, we assume that the signature is acyclic, all assumptions of 
Proposition~\ref{thm:basic} are matched:
\end{comment}

\begin{proposition}\label{prop:mc} A DB theory
	 $T$ has decidable constraint satisfiability problem and admits a model completion in  case it is axiomatized by finitely many universal one-variable formulae and $\Sigma$ is acyclic. 
	 %\begin{itemize}
	 % \item[{\rm (i)}] $\Sigma$ is acyclic and $T$ has the amalgamation property;
	 % \item[{\rm (ii)}] $T$ is empty or contains any subset of the above axioms~\eqref{eq:null}-\eqref{eq:null2}.\footnote{Da controllare.  }
	 %\end{itemize}
         %
	\end{proposition}
%
%	\begin{proof}
%	 See the Appendix.
%	\end{proof}

\begin{proof} 

First, notice that, in case $\Sigma$ is acyclic, $T$ has the finite model property. In fact, if $T:=\emptyset$, then congruence closure ensures that the finite model property holds and decides  constraint satisfiability in time $O(n\log n)$. Otherwise, we reduce the argument to the Herbrand Theorem.  Indeed, suppose to have a finite set $\Phi$ of universal formulae
%silvio 15/2
and let $\phi(\ux)$ be the constraint we want to test for satisfiability. Replace the variables $\ux$ with free constants $\ua$. Herbrand Theorem states that $\Phi\cup \{\phi(\ua)\}$ has a model iff the set of ground $\Sigma^{\ua}$-instances of $\Phi\cup \{\phi(\ua)\}$ has a model. These ground instances are finitely many by acyclicity, so we can reduce to the case where $T$ is empty. Hence,  the constraint satisfiability problem for $T$ is decidable.

%silvio 15/2
To show the existence of a model completion,
we freely take inspiration from an analogous result in~\cite{wheeler,LIP}.  We preliminarily show that $T$ is amalgamable. Then, for a suitable choice of $\psi$ suggested by the acyclicity assumption, the amalgamation property will be used to prove the validy of the condition {\rm (ii)} of Lemma~\ref{lem:cover}: this fact (together with condition {\rm (i)} of Lemma~\ref{lem:cover} ) yields that $T$ has a model completion which is axiomatized by the infinitely many sentences (\ref{eq:qe}).

Let $\cM_1$ and $\cM_2$ two models of $T$ with a submodel $\mathcal{M}_0$ of $T$ in common (we suppose for simplicity that $|\cM_1|\cap|\cM_2|=|\cM_0|)$. We define a $T$-amalgam $\cM$ of $\cM_1, \cM_2$ over $ \cM_0$ as follows (we use in an essential way the fact that $\Sigma$ contains only \textit{unary} function symbols and $n$-ary relation symbols). 
Let the support of $\cM$ be the set-theoretic union of the supports of $\cM_1$ and $\cM_2$, i.e. $|\cM|:= |\cM_1|\cup |\cM_2|$. $\cM$ has a natural $\Sigma$-structure inherited by the $\Sigma$-structures $\cM_1$ and $\cM_2$.
%silvio 15/2 ho rifatto questo pezzo
For every function symbol $f$ in $\Sigma$, we define, for each $m_i\in |\cM_i|$ ($i=1,2$), $f^{\cM}(m_i):=f^{\cM_1}(m_i)$, i.e. the interpretation of $f$ in $\cM$ is the restriction of the interpretation of $f$ in $\cM_i$ for every element $m_i\in |\cM_i|$. This is well-defined since, for every $a\in |\cM_1|\cap|\cM_2|=|\cM_0|$, we have that $f^{\cM}(a):=f^{\cM_1}(a)=f^{\cM_0}(a)=f^{\cM_2}(a)$.  For every $n$-ry relation symbol $R$, we let $R^{\cM}$ be the union of $R^{\cM_1}\cup R^{\cM_2}$; notice that if the tuple $\ua$ belongs to $ R^{\cM}$, then we must have that  $\ua\in |\cM_i|^n$ and that $\ua\in R^{\cM_i}$ for either $i=1$ or $i=2$ (or for both, but this happens just in case $\ua\in |\cM_0|^n=|\cM_1|^n\cap |\cM_2|^n$, because $\cM_0$ is a substructure of $\cM_1$ and $\cM_2$).
\begin{comment}
for every function symbol $f$ in $\Sigma$, we define, for each $m_i\in |\cM_i|$ ($i=1,2$), $f^{\cM}(m_i):=f^{\cM_1}(m_i)$, i.e. the interpretation of $f$ in $\cM$ is the restriction of the interpretation of $f$ in $\cM_i$ for every element $m_i\in |\cM_i|$; for every $n$-ary relation symbol $R$ in $\Sigma$, we define, for each $n$-uple $\underline{m}_i$ of elements from $|\cM_i|$ ($i=1,2$), that $\underline{m}_i\in R^{\cM}$ holds iff $\underline{m}_i\in R^{\cM_i}$. This is well-defined since, for every $a\in |\cM_1|\cap|\cM_2|=|\cM_0|$, we have that $f^{\cM}(a):=f^{\cM_1}(a)=f^{\cM_0}(a)=f^{\cM_2}(a)$ and, for every $n$-uple $\ua$ of elements from $|\cM_0|$, we have that $ \ua \in R^{\cM}$ iff $\ua \in R^{\cM_1}$ iff $\ua\in R^{\cM_0}$ iff $\ua\in R^{\cM_2}$. 
\end{comment}
It is clear that $\cM_1$ and $\cM_2$ are substructures of $\cM$, and their inclusions agree on $\cM_0$.

We show that the $\Sigma$-structure $\cM$, as defined above, is a model of $T$. By hypothesis, $T$ is axiomatized by universal one-variable formulae: so, we can consider $T$ as a theory formed by axioms $\phi$ which are universal closures of clauses with just one variable, i.e. $\phi:=\forall x (A_1(x)\wedge...\wedge A_n(x)\rightarrow B_1(x)\vee...\vee B_m(x))$, where $A_j$ and $B_k$ ($j=1,...,n$ and $k=1,...,m$) are atoms.

We show that $\cM$ satisfies all such formulae $\phi$. In order to do that, suppose that, for every $a\in |\cM|$, $\cM\models A_j(a)$ for all $j=1,...,n$. If $a\in |\cM_i|$, then $\cM\models A_j(a)$ implies $\cM_i\models A_j(a)$, since $A_j(a)$ is a ground formula. Since $\cM_i$ is model of $T$ and so $\cM_i\models \phi$, we get that $\cM_i\models B_k(a)$ for some $k=1,...,m$, which means that $\cM\models B_k(a)$, since $B_k(a)$ is a ground formula. Thus, $\cM\models\phi$ for every axiom $\phi$ of $T$, i.e. $\cM\models T$ and, hence, $\cM$ is a $T$-amalgam of $\cM_1, \cM_2$ over $ \cM_0$, as wanted.
	
Now, given a primitive formula $\exists \ue \,\phi(\ue,\uy)$, we find a suitable $\psi(\uy)$ such that conditions {\rm (i)} and {\rm (ii)} of Lemma~\ref{lem:cover} hold. We define $\psi(\uy)$ as the conjunction of the set of all quantifier-free $\chi(\uy)$-formulae such that $\phi(\ue,\uy)\rightarrow \chi(\uy)$ is a logical consequences of $T$ (they are finitely many - up to $T$-equivalence - because $\Sigma$ is acyclic). By definition, clearly we have that {\rm (i)} of Lemma~\ref{lem:cover} holds.

 We show that also condition {\rm (ii)} of Lemma~\ref{lem:cover} is satisfied. Let $\cM$ be a model of $T$ such that $\cM\models \psi(\ua)$ for some tuple of elements $\ua$ from the support of $\cM$. Then, consider the $\Sigma$-substructure $\cM[\ua]$ of $\cM$ generated by the elements $\ua$: this substructure is finite (since $\Sigma$ is acyclic), it is a model of $T$ and we trivially have that $\cM[\ua]\models \psi(\ua)$, since $\psi(\ua)$ is a ground formula. In order to prove that there exists an extension $\cN^{\prime}$ of $\cM[\ua]$ such that $\cN\models \exists \ue\,\phi(\ue,\ua)$, it is sufficient to prove (by the Robinson Diagram Lemma) that the $\Sigma^{|\cM[\ua]|\cup\{\ue\}}$-theory $\Delta(\cM[\ua])\cup\{\phi(\ue,\ua)\}$ is $T$-consistent. For reduction to absurdity, suppose that the last theory is $T$-inconsistent. Then, there are finitely many literals $l_1(\ua),..., l_m(\ua)$ from $\Delta(\cM[\ua])$ (remember that $\Delta(\cM[\ua])$ is a finite set of literals since $\cM[\ua]$ is a finite structure) such that $\phi(e,\ua)\models_T \neg(l_1(\ua)\wedge...\wedge l_m(\ua))$. Therefore, defining $A(\ua):=l_1(\ua)\wedge...\wedge l_m(\ua)$, we get that $\phi(e,\ua)\models_T \neg A(\ua)$, which implies that $\neg A(\ua)$ is one of the $\chi(\uy)$-formulae appearing in $\psi(\ua)$. Since $\cM[\ua]\models \psi(\ua)$, we also have that $\cM[\ua]\models \neg A(\ua)$, which is a contraddiction: in fact, by definition of diagram, $\cM[\ua]\models A(\ua)$ must hold. Hence, there exists an extension $\cN^{\prime}$ of $\cM[\ua]$ such that $\cN^{\prime}\models\exists \,\ue\phi(\ue,\ua)$.
Now, by amalgamation property, there exists a $T$-amalgam $\cN$ of $\cM$ and $\cN^{\prime}$ over $\cM[\ua]$: clearly, $\cN$ is an extension of $\cM$ and, since $\cN^{\prime}\hookrightarrow \cN$ and $\cN^{\prime}\models\exists \ue\,\phi(\ue,\ua)$, also $\cN\models\exists \ue\,\phi(\ue,\ua)$ holds, as required.
%Show that $T$ has amalgamation. Then to use the Lemma, let $\psi$ be the conjunction of the set of all quantifier-free $\chi(\uy)$-formulae which are logical consequences modulo $T$ of $\phi(x, \uy) $ (they are finitely many - up to $T$-equivalence - because $\Sigma$ is acyclic). Then show Lemma~\ref{lem:qe}(ii)  by using 
%restriction to the substructure generated by the $\ua$, diagrams and compactness (this is possible via amalgamation).
%$\hfill\dashv$
\end{proof}

Since acyclicity of $\Sigma$ yields local finiteness, we immediately get as a Corollary the decidability of safety problems for transitions systems based on DB schema satisfying the hypotheses of the above theorem.  
	
 %\begin{corollary} If the  DB-schema $\tup{\Sigma,T}$ is such that $T$ is axiomatized by universal one-variable formulae and $\Sigma$ is acyclic, then the backward search algorithm of Proposition~\ref{thm:basic} decides safety problems for transition systems of the kind  
 %        $\cS ~=~\tup{\Sigma, T, \ux, \iota(\ux), \tau(\ux, \ux')}$.
 %\end{corollary}

\section{Covers via Constrained Superposition}\label{sec:SC}

Of course, a model completion  may not exist at all; Proposition~\ref{prop:mc} shows that it exists in case $T$ is a DB theory axiomatized by universal one-variable formulae and $\Sigma$ is acyclic. The second hypothesis is unnecessarily restrictive and the 
algorithm for quantifier elimination suggested by the proof of Proposition~\ref{prop:mc} is highly impractical: 
for this reason we are trying a different approach.
In this section, we drop the acyclicity hypothesis and examine the case where the theory $T$ is empty and the signature $\Sigma$ may contain function symbols of any arity. Covers in this context were shown to exist already in~\cite{GM}, using an algorithm that, very roughly speaking, determines all the conditional equations that can be derived concerning the nodes of the congruence closure graph. 
An algorithm for the generation of interpolants, still relying on congruence closure~\cite{kapurCC} and similar to the one presented in~\cite{GM}, is supplied in~\cite{kapur}.

We follow a different plan and we want to produce covers (and show that they exist)  using \emph{saturation-based theorem proving}. The natural idea to proceed in this sense is to take the matrix $\phi(\ue,\uy)$ of the primitive formula $\exists \ue\, \phi(\ue, \uy)$ we want to compute the cover of: this is a conjunction of literals, so we consider each variable as a free constant, we saturate the corresponding set of ground literals and finally we output the literals involving only the $\uy$. For saturation, one can use any version of the superposition calculus~\cite{NR}. This procedure  however for our problem  is not sufficient. As a trivial counterexample consider the primitive formula $\exists e \,(R(e, y_1)\wedge \neg R(e, y_2))$: the set of literals $\{ R(e, y_1), \neg R(e, y_2)\}$ is saturated (recall that we view $e,y_1,y_2$ as constants), however the formula has a non-trivial cover  $y_1\neq y_2$ which is \emph{not} produced by saturation. If we move to signatures with  function symbols, the situation is even worse: the set of literals $\{ f(e,y_1)=y'_1, f(e,y_2)= y'_2\}$ is saturated but the formula $\exists e\, (f(e,y_1)=y'_1\wedge f(e,y_2)= y'_2)$ has the \emph{conditional equality} $y_1=y_2\to y'_1=y'_2$ as  cover. \emph{Disjunctions of disequations} might also arise: the cover 
of $\exists e\, h(e,y_1,y_2)\neq h(e, y'_1,y'_2)$ (as well as the cover of $\exists e\, f(f(e,y_1),y_2)\neq f(f(e,y_1'),y_2')$, see Example~\ref{ex:disj} below)
is $y_1\neq y'_1\vee y_2\neq y'_2$.~\footnote{ This example points out a problem that needs to be fixed in the algorithm presented in~\cite{GM}: %,kapur}:
 that algorithm in fact outputs only equalities, conditional equalities and single disequalities, so it cannot correctly handle this example.
}

Notice that our problem is different from the problem of producing ordinary quantifier-free interpolants via saturation 
based theorem proving~\cite{voronkov}: for ordinary Craig interpolants, we have 
as input \emph{two} quantifier-free formulae $\phi(\ue, \uy), \phi'(\uy, \uz)$ such that $\phi(\ue,\uy)\to \phi'(\uy, \uz)$ is valid; here we have a \emph{single} formula
$\phi(\ue, \uy)$ in input and we are asked to find an interpolant which is good \emph{for all possible} $\phi'(\uy, \uz)$ such that $\phi(\ue,\uy)\to \phi'(\uy, \uz)$ is valid. Ordinary interpolants can be extracted from a refutation of $\phi(\ue,\uy)\wedge \neg \phi'(\uy, \uz)$, here we are not given any refutation at all (and we are not even supposed to find one).

 What we are going to show is that, nevertheless, saturation via superposition can be used to produce covers, if suitably adjusted.
In this section we consider signatures with $n$-ary function symbols (for all $n\geq 1$). For simplicity, 
we omit $n$-ary relation symbols
(you can easily handle them 
by rewriting $R(t_1, \dots, t_n)$ 
as $R(t_1, \dots, t_n)=true$, as customary in the paramodulation literature~\cite{NR}).

We are going to compute the cover of a primitive formula $\exists \ue\,\phi(\ue, \uy)$ to be fixed for the remainder of this section.
We call variables $\ue$ \emph{existential} and variables $\uy$ \emph{parameters}. By applying abstraction steps, we can assume that $\phi$ is \emph{primitive flat}. i.e. that it is  a conjunction of 
\emph{$\ue$-flat literals}, defined below. [By an abstraction step we mean replacing $\exists \ue \,\phi$ with $\exists \ue\, \exists e' (e'= u\wedge \phi')$, where 
$e'$ is a fresh variable and $\phi'$ is obtained from $\phi$ by replacing some occurrences of a term $u(\ue, \uy)$ by $e'$].

A term or a formula are said to be $\ue$-free iff the existential variables do not occur in it. An $\ue$-flat term is an $\ue$-free term $t(\uy)$ or a variable from $\ue$ or again it is of the kind $f(u_1, \dots, u_n)$, where $f$ is a function symbol and $u_1, \dots, u_n$ are 
$\ue$-free terms or variables from $\ue$. An $\ue$-flat literal is a literal of the form
\[
t=a, 
%silvio latest
%mi pare inutile mettere a=b perchè è già incluso in t=a 
%\quad a=b, 
\quad a\neq b
\]
where $t$ is an $\ue$-flat term and $a,b$ are either $\ue$-free terms or variables from $\ue$.

We assume the reader is familiar with standard conventions used in rewriting and paramodulation literature: in particular $s_{\vert p}$ denotes the subterm of $s$ in position $p$ and $s[u]_p$ denotes the term obtained from $s$ by replacing $s_{\vert p}$ with $u$. 
We use $\equiv$ to indicate coincidence of syntactic expressions (as strings) to avoid confusion with equality symbol; when we write   equalities like $s=t$ below, we may mean both $s=t$ or $t=s$ (an equality is seen as a multiset of two terms).
For information on reduction ordering, see for instance~\cite{BaNi98}.

We first replace variables $\ue=e_1,\dots,e_{n}$ and $\uy= y_1, \dots, y_{m}$ by free
constants - we keep the names $e_1,\dots,e_{n}, y_1, \dots, y_{m}$
for these constants. 	 
Choose a reduction ordering $>$ total for ground terms such that $\ue$-flat literals $t=a$ are always oriented from left to right in the following two cases: (i) $t$ is not $\ue$-free and $a$ is $\ue$-free; (ii) $t$ is not $\ue$-free, it is not equal to any of the $\ue$ and $a$ is a variable from $\ue$. To obtain such properties, one may for instance choose a suitable Knuth-Bendix ordering taking weights in some transfinite ordinal, see~\cite{LudW07}.

Given two $\ue$-flat terms $t, u$, we indicate with $E(t,u)$ the following procedure:
\begin{compactitem}[$\bullet$]
 \item $E(t,u)$ fails if $t$ is $\ue$-free and 
 $u$ is not $\ue$-free
 (or vice versa);
 %silvio latest 
 %va aggiunto il caso (ovvio) di fallimento quando ho e_i= termine che comincia con una funzione
 \item $E(t,u)$ fails if $t\equiv e_i$ and 
 (either $t\equiv f(t_1, \dots, t_k)$ or
 $u\equiv e_j$ for $i\neq j$);
 \item $E(t,u)=\emptyset$ if $t\equiv u$;
 \item $E(t,u)=\{t=u\}$ if $t$ and $u$ are different but both $\ue$-free;
 \item $E(t,u)$ fails if none of $t,u$ is $\ue$-free, $t\equiv f(t_1,\dots, t_k)$ and $u\equiv g(u_1,\dots, u_l)$ for $f\not\equiv g$;
 \item $E(t,u)=E(t_1,u_1)\cup \cdots \cup E(t_k,u_k)$ if none of $t,u$ is $\ue$-free, $t\equiv f(t_1,\dots, t_k)$, $u\equiv f(u_1,\dots, u_k)$
 %sg 30/4/19 aggiunta osservazione (che era implicita) come chiesto da un referee
 and none of the $E(t_i, u_i)$ fails.
\end{compactitem}
Notice that, whenever $E(t,u)$ succeeds, the formula $\bigwedge E(t,u)\to t=u$ is universally valid.
The definition of $E(t,u)$ is motivated by the next lemma.
 \begin{lemma}\label{lem:E}
 Let $R$ be a convergent (i.e. terminating and confluent) ground rewriting system, whose rules consist of $\ue$-free terms. Suppose that $t$ and $u$ are $\ue$-flat terms with the same $R$-normal form. Then $E(t,u)$ does not fail and all pairs from $E(t,u)$ have the same $R$-normal form as well. 
\end{lemma}

\begin{proof}
 This is due to the fact that if $t$ is not $\ue$-free, no $R$-rewriting is possible at root position because rules from $R$ are $\ue$-free.
\end{proof}

In the following, we handle \emph{constrained} ground flat literals of the form $L\,\Vert\, C$ where $L$ is a ground flat literal and $C$ is a conjunction of ground equalities among $\ue$-free terms. The logical meaning of $L\,\Vert\, C$ is the Horn clause $\bigwedge C\to L$.

%aggiunto il pezzo e rsicritto il precedente
In the literature, various calculi with constrained clauses were considered, starting e.g. from the  non-ground constrained versions of the Superposition Calculus of~\cite{basic,basic1}. The calculus we propose here is inspired by such versions and it has close similarities with a subcase of hierarchic superposition calculus~\cite{Uwe94}, or rather to its "weak abstraction" variant from~\cite{Uwe13} (we thank an anonymous referee for pointing out this connection).

The rules of our \emph{Constrained Superposition Calculus} follow;
%(they are loosely reminiscent of non-ground constrained versions of the Superposition Calculus known from the literature~\cite{basic,basic1});
each rule applies provided the $E$ subprocedure called by it does not fail. The symbol $\bot$ indicates the empty clause.  
Further explanations and restrictions to the calculus are given in the Remarks below.

\noindent\begin{tabularx}{\textwidth}{cCl}
\begin{tabular}{@{}c@{}}
\bfseries Superposition Right\\
\bfseries (Constrained)
\end{tabular}
&
$\begin{array}{c}
l=r~\Vert~ C\quad\quad s=t ~\Vert~ D
\\
\hline  s[r]_p=t ~\Vert~ C\cup D \cup E(s_{\vert p}, l)
\\
\end{array}$
&
if $l>r$ and 
%silvio 15/2
$s > t$
%$s_{\vert p}\geq t$
\\[12pt]
\begin{tabular}{@{}c@{}}
\bfseries Superposition Left\\
\bfseries (Constrained)
\end{tabular}
&
$\begin{array}{c}
l=r~\Vert~ C\quad\quad s\neq t ~\Vert~ D
\\
\hline  s[r]_p\neq t ~\Vert~ C\cup D \cup E(s_{\vert p}, l)
\\
\end{array}$
&
if $l>r$ and 
%silvio 15/2
$s > t$
%$s_{\vert p}\geq t$
\\[12pt]
\begin{tabular}{@{}c@{}}
\bfseries Reflexion\\
\bfseries (Constrained)
\end{tabular}
&
$\begin{array}{c}
t\neq u~\Vert~ C
\\
\hline  \bot ~\Vert~ C\cup E(t, u)
\\
\end{array}$
&
\\[12pt]
\begin{tabular}{@{}c@{}}
\bfseries Demodulation\\
\bfseries (Constrained)
\end{tabular}
&
$\begin{array}{c}
L~\Vert~ C, \quad \quad l=r~\Vert D
\\
\hline  L[r]_p ~\Vert~ C
\\
\end{array}$
&
\begin{tabular}{@{}rl@{}}
\text{if}&
$l>r$, $L_{\vert p}\equiv l$\\
&and $C\supseteq D$
\end{tabular}
\\
\end{tabularx}

\begin{remark}
 The first three rules are inference rules: 
 they are non-deterministic\-ally selected for application, until no rule applies anymore.
 The selection strategy for the rule to be applied is not relevant for the correctness and completeness of the algorithm (some variant of a `given clause algorithm' can be applied). An inference rule \emph{is not applied in case one  premise  is $\ue$-free} (we have no reason to apply inferences to $\ue$-free premises, since we are not looking for a refutation).
\end{remark}

\begin{remark}
 The Demodulation rule is a simplification rule: 
 its application not only adds the conclusion to the current set of constrained literals, but it also removes the first premise. It is easy to see (e.g., representing literals as multisets of terms and extending the total reduction ordering to multisets), that one cannot have an infinite sequence of consecutive applications of Demodulation rules.
 \end{remark}

\begin{remark}
 The calculus takes
$\{L\Vert \emptyset ~\mid~ L$ is a flat literal from the matrix of $\phi\}$ as the initial set of constrained literals. It terminates when a \emph{saturated} set of constrained literals is reached. We say that $S$ is saturated iff every constrained literal that can be produced by an inference rule, after being 
exhaustively simplified via Demodulation, is already in $S$ (there are more sophisticated notions of `saturation up to redundancy' in the literature, but we do not need them). When it reaches a saturated set $S$, the algorithm outputs the conjunction of the clauses $\bigwedge C\to L$, varying $L\,\Vert \, C$ among the $\ue$-free constrained literals from $S$.
\end{remark}

 \noindent
We need some rule application policy to ensure termination: without any such policy,  
 a set like $\{e=y\,\Vert\, \emptyset, f(e)=e\Vert \,\emptyset\}$ may produce by Right Superposition the infinitely many literals (all oriented from right to left) $f(y)=e\,\Vert\, \emptyset$, $f(f(y))=e\,\Vert\, \emptyset$, $f(f(f(y)))=e\,\Vert\, \emptyset$, etc. The next Remark explains the policy we follow.

 \begin{remark}\label{rem:simpl}
 First, we apply Demodulation \emph{only in case the second premise is of the 
 kind $e_j=t(\uy)\, \Vert D$, where $t$ is $\ue$-free}. 
 Demodulation rule is applied with higher priority with respect to the inference rules. Inside all possible applications of Demodulation rule, we give priority to the 
 applications where \emph{both premises  have the form} $e_j=t(\uy)\, \Vert D$ (for the same $e_j$ but with possibly different $D$'s - the $D$ from the second premise being included in the $D$ of the first).
 In case we have two constrained literals of the kind $e_j=t_1(\uy)\, \Vert D$,  $e_j=t_2(\uy)\, \Vert D$ inside our current set of constrained literals (notice that the $e_j$'s and the $D$'s here are the same), among the two possible applications of  the Demodulation rule,
 we apply the rule that keeps  the smallest $t_i$.  Notice that in this way two different constrained literals cannot simplify each other.
\end{remark}

We say that a constrained literal $L\, \Vert C$ belonging to a set of constrained literals $S$ is \emph{simplifiable in $S$} iff it is possible to apply (according to the above policy) a Demodulation rule removing it.
A first effect of our policy is: %described in the following

\begin{lemma}\label{lem:simplif}
 If a constrained literal 
 $L\,\Vert \, C$
 is simplifiable in $S$, then after applying to $S$ any sequence of rules, it remains simplifiable until it gets removed. 
 After being removed, if it is regenerated, it is still simplifiable and so it is eventually removed again.
\end{lemma}

\begin{proof}
%sg 30/4/19 riscritta la prova in dettaglio come nel rebuttal
Suppose that $L\,\Vert \, C$ can be simplified by $e=t\,\Vert\, D$ and suppose that a rule is applied to the current set of constrained literals. Since there are simplifiable constrained literals, that rule cannot be an inference rule by the priority stated in Remark~\ref{rem:simpl}. For simplification rules, keep in mind again Remark~\ref{rem:simpl}. If $L\,\Vert \, C$ is simplified, it is removed; if none of $L\,\Vert \, C$ and $e=t\,\Vert\, D$ get simplified, the situation does not change; if $e=t\,\Vert\, D$ gets simplified, this can be done by some $e=t'\Vert\,D'$, but then $L\,\Vert \, C$ is still simplifiable - although in a different way - using $e=t'\Vert\,D'$ (we have that $D'$ is included in $D$, which is in turn included in $C$). Similar observations apply if $L\,\Vert \, C$ is removed and re-generated.
 %This is because no inference rule applies until there is at least a simplifiable literal and 
 % a Demodulation rule removing the literal $e_i=t(\uy)\,\Vert D$  that can simplify $L\,\Vert \, C$ requires another literal that can still simplify it.
\end{proof}

Due to the above Lemma, if we show that a derivation (i.e. a sequence of rule applications) can produce terms only from a finite  set, 
it is clear that when no new constrained literal is produced, saturation is reached.
First notice that

\begin{lemma}\label{lem:flat}
 Every constrained literal $L\,\Vert C$ produced during the run of the algorithm is $\ue$-flat.
\end{lemma}

\begin{proof} The constrained literals from initialization are $\ue$-flat. The Demodulation rule, applied according to Remark~\ref{rem:simpl},
produces an $\ue$-flat literal out of an $\ue$-flat literal. The same happens for the Superposition rules: in fact, since both the terms $s$ and $l$ from these rules  are $\ue$-flat,
a Superposition may take place at root position or may rewrite some $l\equiv e_j$ with $r\equiv e_i$ or with $r\equiv t(\uy)$.
\end{proof}

There are in principle infinitely many $\ue$-flat terms that can be generated out of the  $\ue$-flat terms occurring in $\phi$  (see the above counterexample).
 We show however that  only finitely many $\ue$-flat terms can in fact occur during  saturation and that one can determine in advance the finite set they are taken from. 
 %sg 30/4/19
 %% Ho tolto l'osservazione che segue perchè può essere fuorviante: lo stesso u può essere generato in modi diversi a seconda del vincolo
 %(il commento del referee poteva essere stato fuorviato così)
 %The intuitive reason for that is that in building an $\ue$-flat term $u$ via rewriting rules like $e_i\to t(\uy)$, 
 %\emph{we can avoid to use  twice a  rewriting of the same variable $e_i$ along the same building branch for $u$}. 
 
 To formalize this idea, 
let us  introduce a hierarchy of $\ue$-flat terms. Let $D_0$ be the $\ue$-flat terms occurring in $\phi$ and let $D_{k+1}$ be the set of $\ue$-flat terms obtained by simultaneous rewriting of an $\ue$-flat term from $\bigcup_{i\leq k} D_i$ via rewriting rules of the kind $e_j\to t_j(\uy)$ where the $t_j$ are $\ue$-flat $\ue$-free terms from $\bigcup_{i\leq k} D_i$. The \emph{degree} of an $\ue$-flat term is the minimum $k$ such that it belongs to set $D_k$  (it is necessary to take the minimum because the same term can be obtained in different stages and via different rewritings).\footnote{ Notice that, in the above definition of degree, constraints (attached to the rewriting rules occurring in our calculus) are ignored.}

\begin{lemma}\label{lem:deg}%AGjan
 Let the $\ue$-flat term $t'$ be obtained by a rewriting $e_j\to u(\uy)$ from the $\ue$-flat term $t$; then, if $t$ has degree $k>1$ and $u$ has degree at most $k-1$, we have that $t'$ has degree at most $k$.
\end{lemma}

\begin{proof}
 This is clear, because at the $k$-stage one can directly produce $t'$ instead of just $t$: in fact, all rewriting producing directly $t'$ replace an occurrence of some $e_i$ by an $\ue$-free term, so they are all done in parallel positions. 
 \end{proof}

\begin{proposition}\label{lem:R}
 The saturation of the initial set of $\ue$-flat constrained literals always terminates after finitely many steps.
\end{proposition}
		
\begin{proof} 
  We show that  all $\ue$-flat terms that may occur during saturation  have at most degree $n$ (where $n$ is the cardinality of $\ue$). This shows that the saturation must terminate, because only finitely many terms may occur in a derivation (see the above observations).
 Let the algorithm during saturation  reach  the status $S$; we say that \emph{a constraint $C$ 
 allows the explicit definition of $e_j$ in $S$} iff $S$ contains a constrained literal of the kind $e_j=t(\uy)\, \Vert D$ with $D\subseteq C$. Now we show by mutual induction two facts concerning a constrained literal $L\,\Vert \, C\in S$: 
 \begin{compactenum}[(1)]
  \item if an $\ue$-flat term $u$ of degree $k$ occurs in $L$, then $C$ allows the explicit definition of $k$ different $e_j$ in $S$;
  \item if $L$ is of the kind $e_i=t(\uy)$, for an $\ue$-flat $\ue$-free term $t$ of degree $k$, then either $e_i=t\,\Vert\, C$ can be simplified
  in $S$ or $C$ allows the explicit definition of $k+1$ different $e_j$ in $S$ ($e_i$ itself is of course included among these $e_j$).
 \end{compactenum}
    Notice that (1) is
  sufficient to exclude that any $\ue$-flat term of degree bigger than $n$ can occur in a  constrained literal arising during the saturation process.
  
  We prove (1) and (2) by induction on the length of the derivation   leading to $L\,\Vert \, C\in  S$. Notice that it is sufficient to check that (1) and (2) hold for the first time where $L\,\Vert\,  C\in  S$ because if $C$ allows the explicit definition of a certain variable
  in $S$, it will continue to do so in any $S'$ obtained from $S$ by continuing the derivation (the definition may be changed by the Demodulation rule, but the fact that $e_i$ is explicitly defined is forever). Also, by Lemma~\ref{lem:simplif}, a literal cannot become non simplifiable if it is simplifiable.
  
  (1) and (2) are evident if $S$ is the initial status. To show (1), suppose that $u$ occurs for the first time in $L\,\Vert\,C$ as the effect of the application of a certain rule: 
  we can freely assume that $u$ does not occur in the literals from the premisses of the rule (otherwise induction trivially applies) and that $u$ 
  %sg 30/4/19 aggiunta riga seguente
  of degree $k$
  is obtained by rewriting \emph{in a non-root position} %\footnote{Rewriting in root position a variable does not produce new terms.} 
  some $u'$ occurring in a constrained literal $L'\,\Vert\, D'$ via some $e_j\to t\, \Vert\, D$. This might be the effect of a Demodulation or Superposition in a non-root position (Superpositions in root position do not produce new terms).
  %sg 30/4/19
  %CAMBIATO QUI: ho messo la frase che abbiamo usato nel rebuttal
  %%%we suppose that u occurring in L||C of degree k is obtained from previously produced u' occurring in L'||D' via e->t || D. 
  If $u'$ has degree $k$, then by induction $D'$ contains the required $k$ explicit definitions, and we are done because $D'$ is included in $C$. If $u'$ has lower degree, then $t$ must have degree at least $k-1$ (otherwise $u$ does not reach degree $k$ by Lemma~\ref{lem:deg}). Then by induction on (2), the constraint $D$ (also included in $C$) has $(k-1)+1=k$ explicit definitions
  %
  %
  %If $u'$ has degree $k$ and $t$ has lower degree, we can apply the induction hypothesis via Lemma~\ref{lem:deg}; otherwise, we can apply the induction hypothesis for (2) to $e_j\to t\, \Vert D$: in fact in this case, if $t$ has degree $k$, then $u$ has degree at most $k+1$ and, by (2), $D$ carries over $k+1$ explicit definitions  
  (when a constraint $e_j\to t\, \Vert D$ is selected for Superposition 
  or for making Demodulations in a non-root position, it is itself not simplifiable according to the procedure explained in Remark~\ref{rem:simpl}).
  
   To show (2), we analyze the reasons why the non simplifiable constrained literal $e_i=t(\uy)\, \Vert\, C$ is produced  (let $k$ be the degree of $t$).
   Suppose it is produced from $e_i=u'\,\Vert\, C$ via Demodulation with $e_j= u(\uy)\, \Vert\,  D$ (with $D\subseteq C$) in a non-root position; if $u'$ 
   has degree at least $k$,
    we apply induction for (1) to $e_i=u'\,\Vert\, C$: by such induction hypotheses, we get $k$ explicit definitions
   in $C$ and we can add to them the further explicit definition $e_i=t(\uy)$ (the explicit definitions from $C$ cannot concern $e_i$ because 
   $e_i=t(\uy)\, \Vert\, C$ is not simplifiable).
   Otherwise, $u'$ has degree less than $k$ and $u$ has degree 
   at least $k-1$ by Lemma~\ref{lem:deg} 
   %sg 30/4/19: aggiunta riga seguente
   (recall that $t$ has degree $k$): 
   by induction, $e_j= u\, \Vert\, D$ is not simplifiable (it is used as the active part of a Demodulation in a non-root position, see Remark~\ref{rem:simpl}) and supplies $k$ explicit definitions, inherited by $C\supseteq D$. Note that
    $e_i$ cannot have a definition in $D$, otherwise $e_i=t(\uy)\,\Vert\, C$ would be simplifiable, so with $e_i=t(\uy)\,\Vert\, C$
   we get the required $k+1$ definitions.
   
   The remaining case is when $e_i=t(\uy)\, \Vert\, C$ is produced via Superposition Right. 
   Such a Superposition might be at root or at a non-root position. We first analyse the case of a root position.
   This might be via $e_j=e_i\,\Vert\, C_1$ and
   $e_j=t(\uy)\, \Vert\, C_2$ (with $e_j>e_i$ and $C=C_1\cup  C_2$ because $E(e_j,e_j)=\emptyset$), but in such a case one can easily apply induction. Otherwise, we have a different kind of Superposition at root position: $e_i=t(\uy)\, \Vert\, C$ is obtained from $s=e_i\,\Vert \, C_1$ and 
   $s'=t(\uy)\, \Vert\, C_2$, with $C=C_1\cup C_2\cup E(s,s')$. In this case, by induction for (1), $C_2$ supplies $k$ explicit definitions, to be inherited by $C$. Among such definitions, there cannot be an explicit definition of $e_i$ otherwise $e_i=t(\uy)\, \Vert\, C$ would be simplifiable, so again  we get the required $k+1$ definitions.
   
   In case of a Superposition at a non root-position, we have that $e_i=t(\uy)\, \Vert \, C$ is obtained from $u'=e_i\,\Vert \, C_1$ and 
   $e_j=u(\uy)\, \Vert\, C_2$, with $C=C_1\cup C_2$; here $t$ is obtained from $u'$ by rewriting $e_j$ to $u$. This case is handled similarly to the case where $e_i=t(\uy)\, \Vert \, C$ is obtained via Demodulation rule.
\end{proof}

%sg 30/4/19 aggiunto il pezzo
	Having established termination, we now prove that our calculus computes covers; to this aim, we rely on refutational completeness of unconstrained Superposition Calculus (thus, our technique resembles the technique used \cite{Uwe94,Uwe13} in order to prove refutational completeness of hierarchic superposition,
	although it is not clear whether Theorem~\ref{thm:main} below can be derived from the results concerning 
	hierarchic superposition -
	we are not  just proving refutational completeness and we need to build proper superstructures): 
		
\begin{theorem}\label{thm:main}
Suppose that the above algorithm, taking as input the primitive $\ue$-flat formula $\exists \ue\,\phi(\ue, \uy)$, gives as output the quantifier-free formula $\psi(\uy)$. Then the latter is a cover of $\exists \ue\,\phi(\ue, \uy)$.
\end{theorem}

\begin{proof}
 Let $S$ be the saturated set of constrained literals produced upon termination of the algorithm; let $S=S_1\cup S_2$, where $S_1$ contains the constrained literals in which the $\ue$ do not occur and $S_2$ is its complement.
 Clearly $\exists \ue\,\phi(\ue, \uy)$ turns out to be logically equivalent to
 $$
  \bigwedge_{L\,\Vert\, C\in S_1} (\bigwedge C\to L) \wedge \exists \ue \bigwedge_{L\,\Vert\, C\in S_2} (\bigwedge C\to L)
 $$
so, as a consequence, in view of Lemma~\ref{lem:cover} it is sufficient to show that every model $\cM$ satisfying 
$\bigwedge_{L\,\Vert\, C\in S_1} (\bigwedge C\to L)$ 
 via an assignment $\cI$ to the variables $\uy$ can be embedded into a model $\cM'$ such that for a suitable extension $\cI'$ of $\cI$ to the variables $\ue$ we have that $(\cM', \cI')$ satisfies also $\bigwedge_{L\,\Vert\, C\in S_2} (\bigwedge C\to L)$.

Fix $\cM, \cI$ as above. The  diagram $\Delta(\cM)$ of $\cM$ is obtained as follows. We take one free constant for each element of the support of $\cM$ (by L\"owenheim-Skolem theorem you can keep $\cM$ at most countable, if you like)
and we put in $\Delta(\cM)$ all the literals 
of the kind $f(c_1, \dots, c_k)=c_{k+1}$ and $c_1\neq c_2$ which are 
true in $\cM$ (here the $c_i$ are names for the elements of the support of $\cM$). Let $R$ be the set of ground equalities of the form $y_i=c_i$, where
$c_i$ is the name of $\cI(y_i)$. Extend our reduction ordering in the natural way (so that  $y_i=c_i$ 
and $f(c_1, \dots, c_k)=c_{k+1}$
are oriented from left to right).
Consider now the set of clauses 
\begin{equation}\label{eq:clauses}
\Delta(\cM)~\cup~ R~\cup~ \{\bigwedge C\to L\mid (L\,\Vert\, C)\in S\}
\end{equation}
(below, we distinguish the positive and the negative literals of $\Delta(\cM)$ so that  $\Delta(\cM)=\Delta^+(\cM)\cup \Delta^-(\cM)$). 
We want to saturate the above set in the standard Superposition Calculus. Clearly the rewriting rules in $R$, used as reduction rules, replace everywhere 
$y_i$ by $c_i$ inside the clauses of the kind $\bigwedge C\to L$. 
At this point, the negative literals from the equality constraints all disappear: if they are true in $\cM$, they $\Delta^+(\cM)$-normalize to trivial equalities $c_i
= c_i$ (to be eliminated by standard reduction rules) and if they are false in  $\cM$ they become part of  clauses
subsumed by true inequalities from $\Delta^-(\cM)$.
%of the form $c_1\neq c_2$. 
Similarly all the $\ue$-free literals not coming from $\Delta(\cM)\cup R$ get removed. 
Let $\tilde S$ be the set of survived  literals
involving the $\ue$ (they are not constrained anymore and they are $\Delta^+(\cM)\cup R$-normalized): we show that they cannot produce new clauses. Let in fact $(\pi)$ be an inference 
from the Superposition Calculus~\cite{NR} applying to them. Since no superposition with $\Delta(\cM)\cup R$ is possible, this inference must involve only literals from $\tilde S$; suppose it produces a literal 
$\tilde L$ from
the literals $\tilde L_1, \tilde L_2$ (coming via $\Delta^+(\cM)\cup R$-normalization from $L_1\,\Vert\, C_1\in S$ and $L_2\,\Vert\, C_2\in S$) as parent clauses. Then, by Lemma~\ref{lem:E}, our constrained inferences produce a constrained literal $L\,\Vert\, C$ such that the clause $\bigwedge C\to L$ normalizes to $\tilde L$ via $\Delta^+(\cM)\cup R$. Since $S$ is saturated, the constrained literal $L\,\Vert\, C$, after simplification, belongs to $S$. Now simplifications via our Constrained Demodulation and $\Delta(\cM)^+\cup R$-normalization commute (they work at parallel positions, see Remark~\ref{rem:simpl}), so
the inference $(\pi)$ is redundant because $\tilde L$
simplifies to  a literal already in $\tilde S\cup \Delta(\cM)$.
%cannot have literals from  
% (and cannot be
%superposed with those from $\Delta(\cM)$): in fact, new possible superpositions (due to $\ue$-free term identifications coming from $R$-normalization) 
%where already performed by our algorithm (see %Lemma~\ref{lem:R} and  the form of our constrained Superposition rules).

Thus the set of clauses~\eqref{eq:clauses} saturates without producing the empty clause. By the completeness theorem of the Superposition Calculus~\cite{HR,BG,NR}
it has a model $\cM'$. This $\cM'$ by construction fits our requests by Robinson Diagram Lemma. 
\end{proof}

Theorem~\ref{thm:main} also proves the existence of the model completion of (EUF).%AG:21feb
\begin{comment}
Consider now the set of clauses 
\begin{equation}\label{eq:clauses}
\Delta(\cM)~\cup~ R~\cup~ \{\bigwedge C\to L\mid (L\,\Vert\, C)\in S\}.
\end{equation}
We want to saturate it in the standard Superposition Calculus. Clearly the rewrite rules in $R$, used as reduction rules, replace everywhere 
$y_i$ by $c_i$ inside the clauses of the kind $\bigwedge C\to L$. 
At this point, the negative literals from the equality constraints all disappear: if they are true in $\cM$, they normalize to trivial equalities $c_i
= c_i$ (to be eliminated by standard reduction rules) and if they are false in  $\cM$ they become part of  clauses
subsumed by true inequalities of the form $c_1\neq c_2$. Similarly all the $\ue$-free literals  not coming from $\Delta(\cM)\cup R$ get removed. The literals involving the $\ue$ (not constrained anymore) cannot produce new clauses (and cannot be
superposed with those from $\Delta(\cM)$): in fact, new possible superpositions (due to $\ue$-free term identifications coming from $R$-normalization) where already performed by our algorithm (see Lemma~\ref{lem:R} and  the form of our constrained Superposition rules).

Thus the set of clauses~\eqref{eq:clauses} saturates without producing the empty clause. By the completeness theorem of the Superposition Calculus~\cite{HR,BG}
it has a model $\cM'$. This $\cM'$ by construction fits our requests by Robinson Diagram Lemma.
\end{comment}

\begin{example}
\label{ex:disj}
 We compute the cover of  the primitive formula $\exists e\, f(f(e,y_1),y_2)\neq f(f(e,y_1'),y_2')$.   Flattening gives the set of literals
 $$\{~f(e,y_1)=e_1, ~f(e_1,y_2)=e'_1,~f(e,y_1')=e_2,~ f(e_2,y_2')=e'_2,~ e'_1\neq e'_2 ~\}~~.$$
 Superposition Right produces the constrained literal $ e_1= e_2  \,\Vert\, \{y_1=y'_1\}$; supposing that we have $e_1> e_2$, Superposition Right
 gives first $f(e_2,y_2)=e'_1 \,\Vert\, \{y_1=y'_1\}$ and then also $ e'_1=e'_2\,\Vert\, \{y_1=y'_1, y_2=y'_2\}$. Superposition Left and Reflexion now produce $ \bot  \,\Vert\, \{y_1=y'_1, y_2=y'_2\}$. 
 Thus the clause $y_1=y'_1 \wedge y_2=y'_2 \to \bot$ 
 will be part of the output (actually, this will be the only clause in the output).  
\end{example}

\begin{comment}
 There are easy potential improvement to our calculus.  For instance, there is no need that constraints are conjunctions of equalities, they can be boolean combinations of equalities. Thus one can merge $L\,\Vert\, C$ and $L\,\Vert\,D$ to $L\,\Vert\, C\vee D$. In this way, we can have at most one constraint  per literal.
\end{comment}
%

%\vspace{-3mm}

%\section{One More Example}\label{app:example}

\begin{example}  We add one more example, taken from~\cite{GM}.
 We compute the cover of  the primitive formula $\exists e\, (s_1=f(y_3,e)\land s_2=f(y_4,e)\land t=f(f(y_1,e), f(y_2,e)))$, where $s_1, s_2, t$ are terms in $\uy$. This example is taken from~\cite{GM}.
  Flattening gives the set of literals
 $$\{~f(y_3,e)=s_1, ~f(y_4, e)=s_2,~ f(y_1,e)=e_1, ~f(y_2, e)=e_2,~ f(e_1,e_2)=t~\}~~.$$
 
 \noindent
 Suppose that we have $e>e_1> e_2>t>s_1>s_2>y_1>y_2>y_3>y_4$.
 Superposition Right between the $3$rd and the $4$th clauses produces the constrained $6$th clause $e_1=e_2 \,\Vert\, \{y_1=y_2\}$. From now on, we denote the application of a Superposition Right to the $i$th and $j$th clauses with $R(i,j)$. We list a derivation performed by our calculus:
 
  $$R(3,4)\implies e_1=e_2\,\Vert\, \{y_1=y_2\} \ \ \ \ (6\mbox{th clause})$$
 $$R(1,2)\implies s_1=s_2\,\Vert\, \{y_3=y_4\} \ \ \ \ (7\mbox{th clause})$$
 $$R(5,6)\implies f(e_2,e_2)=t\,\Vert\, \{y_1=y_2\}\ \ \ \ (8\mbox{th clause})$$
  $$R(1,3)\implies e_1=s_1\,\Vert\, \{y_1=y_3\} \ \ \ \ (9\mbox{th clause})$$
 $$R(1,4)\implies e_2=s_1\,\Vert\, \{y_2=y_3\}\ \ \ \ (10\mbox{th clause})$$
  $$R(2,3)\implies e_1=s_2\,\Vert\, \{y_1=y_4\} \ \ \ \  (11\mbox{th clause})$$
   $$R(2,4)\implies e_2=s_2\,\Vert\, \{y_2=y_4\} \ \ \ \ (12\mbox{th clause})$$
   $$R(5, 9)\implies f(s_1,e_2)=t\,\Vert\, \{y_1=y_3\}\ \ \ \ (13\mbox{th clause})$$
   $$R(5,11)\implies f(s_2,e_2)=t\,\Vert\, \{y_1=y_4\}\ \ \ \ (14\mbox{th clause})$$
   $$R(6,9)\implies e_2=s_1\,\Vert\, \{y_1=y_3, y_1=y_2\} \ \ \ \ (15\mbox{th clause})$$
    $$R(6,11)\implies e_2=s_2\,\Vert\, \{y_1=y_2, y_1=y_4\} \ \ \ \ (16\mbox{th clause})$$
 $$R(8,10)\implies f(s_1,s_1)=t\,\Vert\, \{y_1=y_3, y_2=y_3\}\ \ \ \  (17\mbox{th clause})$$
 $$R(8,12)\implies f(s_2,s_2)=t\,\Vert\, \{y_1=y_4, y_2=y_4\} \ \ \ \ (18\mbox{th clause})$$ 
  $$R(13,12)\implies f(s_1,s_2)=t\,\Vert\, \{y_1=y_3, y_2=y_4\}\ \ \ \ (19\mbox{th clause})$$
  $$R(14,10)\implies f(s_2,s_1)=t\,\Vert\, \{y_1=y_4, y_2=y_3\}\ \ \ \  (20\mbox{th clause})$$
  $$R(9,11)\implies s_1=s_2\,\Vert\, \{y_1=y_3, y_1=y_4\}\ \ \ \  (21\mbox{th clause})$$

    %$$\textbf{8\mbox{th clause deleted}}$$
  \vspace{2mm}
  \noindent
  The set of clauses above is saturated. The $7$th, $17$th, $18$th, $19$th and $20$th clauses are exactly the output clauses of~\cite{GM}. The non-simplified clauses that do not appear as output in~\cite{GM} are redundant and they could be simplified by introducing a Subsumption rule as an additional simplification rule of our Constrained Superposition Calculus.
 
 \end{example}

 %Superposition Right produces the constrained literal $ e_1= e_2  \,\Vert\, \{y_1=y'_1\}$; supposing that we have $e_1> e_2$, Superposition Right
 %gives first $f(e_2,y_2)=e'_1 \,\Vert\, \{y_1=y'_1\}$ and then also $ e'_1=e'_2\,\Vert\, \{y_1=y'_1, y_2=y'_2\}$. Superposition Left and Reflexion now produce $ \bot  \,\Vert\, \{y_1=y'_1, y_2=y'_2\}$. 
% Thus the clause $y_1=y'_1 \wedge y_2=y'_2 \to \bot$ 
 %will be part of the output (actually, this will be the only clause in the output). 

\section{Complexity Analysis}
\label{app:complexity}

In the special case where the signature $\Sigma$ contains only unary function symbols, only empty constraints can be generated; 
in case $\Sigma$ contains also relation symbols of arity $n>1$, the only constrained clauses that can be generated have the form $\bot \,\Vert \{t_1=t'_1,\dots, t_{n-1}=t'_{n-1}\}$. %AGfeb
%Also, it is not difficult to see that in such case $\ue$-flat terms can be generated only by the same sequence of rewritings $e_{i_1}\to t_1, \dots, e_{i_k}\to t_k$ (for some $k\leq n$)
%along all building branches in the whole saturation process: 
Also, it is not difficult to see that in a derivation at most one explicit definition $e_i=t(\uy) || \emptyset$ can occur for every $e_i$: as soon as this definition is produced, all occurrences of $e_i$ are rewritten to $t$. This shows that Constrained Superposition computes covers in polynomial time for the empty theory, whenever the signature $\Sigma$ matches the restrictions of Definition~\ref{def:db} for DB schemata. We give here a finer complexity analysis, in order to obtain a quadratic bound.

%We analyze the complexity of our ``constrained SC'' in the
%case of unary functions and $m$-ary relations. 
In this section, \emph{we assume that our signature $\Sigma$ contains only unary function and $m$-ary relation symbols.}
As we already mentioned, the complexity of the saturation becomes polynomial in this restricted case. However, in order to attain an optimized quadratic complexity bound, we need to follow a \emph{different strategy in applying the rules of our constrained superposition calculus} (this different strategy would not be correct for the general case). Thanks to this different strategy, we can make our procedure close to the algorithm of~\cite{GM}: in fact, such algorithm is correct for the case of unary functions and requires only a very minor adjustment for the case of unary functions and $m$-ary relations (the reason why only a very minor adjustment is sufficient is due to the fact that the empty clause is the only constrained clause that is generated during saturation, so that constraints do not become part of a recursive mechanism).

Since relations play a special role in the present restricted context, we prefer to treat them as such, i.e. not to rewrite $R(t_1,\dots, t_n)$ as 
$R(t_1,\dots, t_n)=true$; the consequence is that we need an additional
 Constrained Resolution Rule\footnote{We extend the definition of an $\ue$-flat literal so as to include also the literals of the kind $R(t_1,..,t_n)$ and $\neg R(t_1,..,t_n)$ where the terms $t_i$ are either $\ue$-free terms or variables from $\ue$.}.
We preliminarily notice that when function symbols are all unary, the constraints remain all \emph{empty} during the run of the saturation procedure, except for the case of the newly introduced Resolution Rule below. This fact follows from the observation that
%when a rule applies, the subprocedure $E$ either fails or produces the empty set as a result. This, in turn, implies that starting from a set of literals, no new constraints can be generated. In fact, 
given two terms $u_1$ and $u_2$, procedure $E(u_1,u_2)$ does not fail iff:
\begin{compactenum}[(1)]
\item either $u_1$ and $u_2$ are both terms containing only variables from $\uy$, or  
\item $u_1$ and $u_2$ are terms  that syntactically coincide.
\end{compactenum}
In case (1), $E(u_1,u_2)$ is $\{u_1,u_2\}$ and in case (2), $E(u_1,u_2)$ is $\emptyset$.
In case (1), Superposition Rules are not applicable. To show this, suppose that $u_1 \equiv s_{\vert p}$ and $u_2\equiv l$; then, terms $l$ and $r$ use only variables from $\uy$, and consequently cannot be fed into Superposition Rules, since Superposition Rules are only applied when variables from $\ue$ occur in both premises. Reflexion Rule does not apply too in case (1), because this rule (like any other rule) cannot be applied to an $\ue$-free literal.

Thus, in the particular case of $m$-ary relations and unary functions, the rules of the calculus are the following:
%(the restrictions and the application strategy are the same of the general case):

\noindent\begin{tabularx}{\textwidth}{cCl}
\begin{tabular}{@{}c@{}}
\bfseries Superposition
\end{tabular}
&
$\begin{array}{c}
l=r\quad\quad L
\\
\hline  L[r]_p 
\\
\end{array}$
&
\begin{tabular}[t]{@{}rl@{}}
\text{if}&
 (i) $l>r$;\\
 &(ii) if $L\equiv s=t$ or \\
 &\phantom{(ii) }$L\equiv s\neq t$, then\\
 &\phantom{(ii) }$s>t$ and $p\in Pos(s)$;\\
 &(iii) $E(s_{\vert p},l)$ does not fail.
 \end{tabular}
\\[12pt]
\begin{tabular}{@{}c@{}}
\bfseries Resolution
\end{tabular}
&
$\begin{array}{c}
R(t_1,\dots, t_n) \quad\quad \neg R(s_1,\dots,s_n)
\\
\hline  \bot ~\Vert~ \bigcup_i E(s_i,t_i)
\\
\end{array}$
&
\begin{tabular}[t]{@{}rl@{}}
\text{if}& 
$E(s_i,t_i)$ does not fail\\
& for all $i=1,\dots,n$
\end{tabular}
\\[12pt]
\begin{tabular}{@{}c@{}}
\bfseries Reflexion
\end{tabular}
&
$\begin{array}{c}
t\neq u
\\
\hline  \bot 
\\
\end{array}$
&
\begin{tabular}{@{}rl@{}}
\text{if}&
$E(t,u)$ does not fail
\end{tabular}
\\[12pt]
\begin{tabular}{@{}c@{}}
\bfseries Demodulation
\end{tabular}
&
$\begin{array}{c}
L \quad \quad l=r
\\
\hline  L[r]_p 
\\
\end{array}$
&
\begin{tabular}{@{}rl@{}}
\text{if}&
$l>r$ and $L_{\vert p}\equiv l$\\
\end{tabular}
\\
\end{tabularx}
\vskip 1mm
\begin{comment}
\begin{remark} In case $\Sigma$ contains relation symbols of arity $m>1$, the only constrained clauses that can be generated have the form $\bot \,\Vert \,\{ t_1=t'_1,\dots, t_{n-1}=t'_{n-1}\}$. Also, it is not difficult to see that in a derivation at most one explicit definition $e_i=t(\uy) || \emptyset$ can occur for every $e_i$, since as soon as such an explicit definition is produced, all occurrences of $e_i$ are immediately rewritten as $t$ by Demodulation. 
\end{remark}
\end{comment}
We still restrict the use of our rules to the case where all premises are not $\ue$-free literals; again Demodulation is applied only in the case where $l=r$ is of the kind $e_i=t(\uy)$.
For the order of applications of the Rules, Lemma~\ref{lem:restr} below show that we can apply (restricted) Superpositions, Demodulations, Reflexions and Resolutions in this order and then stop.

An important preliminary observation to obtain such result is that \emph{we do not need to apply Superposition Rules whose left premise $l=r$ is of the kind $e_i=t(\uy)$}: this is because constraints are always empty (unless the constrained clause is the empty clause), so that a Superposition Rule with the left premise $e_i=t(\uy)$ can be replaced by a Demodulation Rule.~\footnote{This is not true in the general case where constraints are not empty, because the Demodulation Rule does not merge incomparable constraints.} If the left premise of Superposition is not of the kind $e_i=t(\uy)$,  then since our literals are $\ue$-flat, it can be either of the kind $e_i=e_j$ (with $e_i> e_j$) or of the kind $f(e_i)= t$. In the latter case $t$ is either $e_k\in \ue$ or it is an $\ue$-free term; for Superposition Left (i.e. for Superposition applied to a negative literal), the left premise can only be $e_i=e_j$,
 because our literals are $\ue$-flat and so negative literals $L$ cannot have a position $p$ such that $L_{\vert p}\equiv f(e_i)$.    

Let $S$ be a set of $\ue$-flat literals with empty constraints; 
%(recall that we assume that all function symbols are unary in this section).
we say that $S$ is 
%\emph{closed for congruence closure} (
\emph{RS-closed} iff it is closed under \emph{Restricted Superposition Rules}, i.e under Superposition Rules whose left premise is not of the kind $e_i=t(\uy)$. In equivalent terms, as a consequence of the above discussion, $S$ is RS-closed 
%for short) 
iff it satisfies the following two conditions:
\begin{compactitem}
 \item if $\{f(e_i)=t, f(e_i)=v\}\subseteq S$, then $t=v\in S$;%\footnote{Since all literals in $S$ are $\ue$-flat, it is clear that $t$ and $u$ are either from $\ue$ or $\ue$-free.}
 \item if $\{e_i=e_j, L\} \subseteq S$ and $e_i>e_j$ and $L_{\vert p}\equiv e_i$, then $L[e_j]_p\in S$.
\end{compactitem} 
Since Restricted Superpositions do not  introduce essentially new terms (newly introduced terms are just rewritings of variables with variables), it is clear that we can make a finite set $S$ of $\ue$-free literals RS-closed in finitely many steps. This can be naively done in time quadratic in the size of the formula. As an alternative, we can apply a \emph{congruence closure algorithm to $S$} and produce a set of $\ue$-free constraints $S'$ which is RS-closed and logically equivalent to $S$: the latter can be done in $O(n\cdot log (n))$-time, as it is well-known from the literature~\cite{NeOp,IC,kapurCC}.

\begin{lemma}\label{lem:restr}
 Let $S$ be a RS-closed set of empty-constrained $\ue$-flat literals.
 Then, to saturate $S$ it is sufficient to first exhaustively apply  the Demodulation Rule, and then Reflexion and Resolution Rules.
\end{lemma}

\begin{proof}
 Let $\tilde S$ be the set obtained from $S$ after having exhaustively applied Demodulation. 
 Notice that the final effect of the reiterated application of Demodulation can be synthetically described by saying that literals in $S$ are rewritten by using some explicit definitions 
 \begin{equation}\label{eq:rew}
 e_{i_1}= t_1(\uy), \dots, e_{i_k}= t_k(\uy)~~.
 \end{equation}
These definitions are either in $S$, or are generated through the Demodulations themselves (we can freely assume that Demodulations are done in appropriate order: first all occurrences of $e_{i_1}$ are rewritten to $t_1$, then  all occurrences of $e_{i_2}$ are rewritten to $t_2$, etc.).\footnote{In addition, if we happen to have, say, two different explicit definitions of $e_{i_1}$ as $e_{i_1}=t_1, e_{i_1}=t'_1$, we decide to use just one of them (and always the same one, until the other one is eventually removed by Demodulation).}
 
 Suppose now that a pair $L, l=r\in \tilde S$ can generate a new literal $L[r]_p$ by Superposition. We know from above that we can limit ourselves to Restricted Superposition, so $l$ is either of the form $e_j$ or of the form $f(e_j)$, where moreover $e_j$ is not among the set $\{e_{i_1}, \dots, e_{i_k}\}$ from~\eqref{eq:rew}.  
 The literals $L$ and $l=r\in \tilde S$ happen to have been obtained from literals $L'$ and $l=r'$ belonging to $S$ by applying the rewriting rules~\eqref{eq:rew}
 (notice that $l$ cannot have been rewritten).
 Since such rewritings must have occurred in positions parallel to $p$ and since $S$ was closed under Restricted Superposition, we must have that $S$ contained the
 literal $L'[r']_p$ that rewrites to $L[r]_p$ by the rewriting  rules~\eqref{eq:rew}. This shows that $L[r]_p$ is already in $\tilde S$ and proves the lemma, because Reflexion and Resolution can only produce the empty clause and no rule applies to the empty clause.
 \end{proof}
 
 \noindent
 Thus the strategy of applying (in this order) 
 \vskip .5mm\noindent
 \centerline{\framebox{Restricted Superposition+Demodulation+Reflexion+Resolution}}
 \par
 \noindent
 always saturates.
 
 To produce an output in optimized format, it is convenient to get it in a dag-like form. This can be simulated via explicit acyclic definitions as follows.
When we write $\mathit{Def}(\ue, \uy)$ (where $\ue, \uy$ are tuples of distinct variables), we mean any flat formula of the kind (let $\ue:=e_1 \dots, e_n$)
$$
\bigwedge_{i=1}^n e_i =t_i
$$
where in the term $t_i$ only the variables $e_1, \dots, e_{i-1}, \uy$ can occur. We shall supply the output in the form 
\begin{equation}\label{eq:dag}
\exists \ue' (\mathit{Def}(\ue', \uy) \wedge \psi(\ue', \uy))
\end{equation}
where the $\ue'$ is a subset of the $\ue$ and $\psi$ is quantifier-free. The \emph{dag-format}~\eqref{eq:dag} is not quantifier-free but can be converted to a quantifier-free formula by unravelling the acyclic definitions of the $\ue'$.

Thus our procedure for computing a cover in dag-format of a primitive formula $\exists \ue\,\phi(\ue, \uy)$ (in case the function symbols of the signature $\Sigma$ are all unary) runs by performing the following steps, one after the other. Let $OUT$ be a quantifier-free formula (initially $OUT$ is $\top$).
\begin{compactenum}[(1)]
\item We preprocess $\phi$ in order to get out of it a RS-closed set $S$ of empty-constrained $\ue$-flat literals.
\item We \emph{mark} the variables $\ue$ in the following way (initially, all variables are unmarked): we scan $S$ and, as soon as we find an equality of the kind $e_i=t$ where all variables from $\ue$ occurring in $t$ are marked, we mark $e_i$. This loop is repeated until no more variable gets marked.
\item If Reflexion is applicable, output $\bot$ and exit.
\item Conjoin  $OUT$ with all literals where, besides the $\uy$, only marked variables occur.
\item For every literal $R(t_1,\dots,e,\dots, t_m)$ that contains at least an unmarked $e$,   scan $S$ until a literal of the type $\neg R(t_1,\dots,e,\dots, t_m)$ is found: then, try to apply Resolution and if you succeed getting $\bot \,\Vert\, \{ u_1=u'_1, \dots, u_m=u'_m\}$ conjoin 
$\bigvee_j u_j\neq u'_j$ to $OUT$.
\item Prefix to $OUT$ a string of existential quantifiers binding all marked variables and output the result.
\end{compactenum}

\noindent
One remark is in order: when running the subprocedures $E(s_i,t_i)$ required by the Resolution Rule  in (5) above, \emph{you must consider all marked variables as part  of the $\uy$} (thus, e.g.
$R(e, t), \neg R(e,v)$ produces $\bot \,\Vert \,\{t=u\}$ if both $t$ and $u$ contain, besides the $\uy$, only marked variables).
\begin{proposition}\label{prop:complex}
  Let $T$ be the theory (EUF) in a signature with unary functions and $m$-ary relation symbols. Consider a primitive
  formula
  $\exists \ue\,\phi(\ue, \uy)$;
   then, the above algorithm returns a cover of $\exists \ue\,\phi(\ue, \uy)$ in dag-format in time 
  $O(n^2)$, where $n$ is the size of $\exists \ue\,\phi(\ue, \uy)$.
\end{proposition}

%\vspace{3mm}
\begin{proof}
The preprocessing step (1) requires an abstraction phase for producing $\ue$-flat literals and a second phase in order to get a RS-closed set:
 the first phase requires linear time, whereas the second one requires $O(n\cdot log(n))$ time. All the remaining steps require linear time, except steps 
 (2) and (5) that requires quadratic time. This is the dominating cost, thus
 the entire procedure requires $O(n^2)$ time.
\end{proof}

Although we do not deeply investigate the problem here, we conjecture that it might be possible to further lower down the above complexity to $O(n\cdot log(n))$.

        \subsection{An extension}\label{subsec:extension}
        
        We consider a useful extension of the above algorithm; 
        let us assume that we have a nonempty theory whose  axioms are~\eqref{eq:null}, namely
        \begin{equation*}
	 \forall x~(x = \nullv \leftrightarrow f(x) = \nullv)
	\end{equation*}
        %\begin{equation}\label{eq:suppl}
        % t(x)= \nullv \leftrightarrow x=\nullv
        %\end{equation}
        %for every term $t$ (here we assume to have many constants $\nullv$, one for every sort). 
        for every function symbol $f$.
        
        One side of the above axiom is equivalent to the ground literal $f(\nullv)=\nullv$ and as such it does not interfer with the completion process  (we just add it to our constraints from the very beginning). 
        
        To accommodate the other side, we need to modify our Calculus. First, we add  to the 
        Constrained Superposition Calculus of Section~\ref{sec:SC}
        the following extra Rule

    \noindent\begin{tabularx}{\textwidth}{cCl}
    \\[3pt]
\hspace{8mm}  \begin{tabular}{@{}c@{}}
\bfseries Inference Rule $Ext(\nullv)$\\
\bfseries (Constrained)
\end{tabular}
\hspace{8mm}
&
$\begin{array}{c}
f(e_j)=u(\uy)\,\Vert\, D
\\
\hline  e_j= \nullv\,\Vert\, D\cup \{u(\uy)=\nullv\}
\\
\end{array}$
&
\\[12pt]
 \end{tabularx}       
    \noindent    
The Rule is sound because $ u(\uy)=\nullv \wedge  f(e_j)=u(\uy) \to   e_j= \nullv$ follows from the axioms~\eqref{eq:null}.
For cover comuptation with our new axioms, we need a restricted version of Paramodulation Rule:

%\noindent
\begin{tabularx}{\textwidth}{cCl}
\\[3pt]
\begin{tabular}{@{}c@{}}
\bfseries Paramodulation \\
\bfseries (Constrained)
\end{tabular}
&
$\begin{array}{c}
e_j=r\,\Vert\, C\quad\quad L\,\Vert\, D
\\
\hline  L[r]_p \, \Vert \, C\cup D
\\
\end{array}$
%&
%\begin{tabular}[t]{@{}rl@{}}
~~(\text{if} $e_j>r ~\& ~L_{\vert p}\equiv e_j$) \\
%&
% (i) $e_j>r$;\\
% &(ii) $L_{\vert p}\equiv l$ \\
% \end{tabular}
&
\\[12pt]
\end{tabularx}
        
 % \vspace{2mm}
  \noindent
Notice that we can have $e_j>r$ only in case $r$ is either some existential variable $e_i$ or  it is an $\ue$-free term $u(\uy)$. Paramodulation Rule (if it is not a Superposition) can only apply to a right member of an equality and such a right member must be $e_j$ itself (because our literals are flat).
Thus the rule cannot introduce new terms and consequently it does not compromize the termination argument of Proposition~\ref{lem:R}.

\emph{The proof of Theorem~\ref{thm:main} can be easily adjusted as follows.} We proceed as in the proof of Theorem~\ref{thm:main}, so  as to obtain the  set $\Delta(\cM)\cup R \cup \tilde S$ which is saturated in the standard (unconstrained) Superposition Calculus . Below, we refer  to the general refutational completeness proof of the Superposition Calculus given in~\cite{NR}. Since we only have unit literals here, in order to produce a model of $\Delta(\cM)\cup R \cup \tilde S$, we can just consider the convergent ground rewriting system $\rightarrow$ consisting of the oriented equalities in $\Delta^+(\cM)\cup R \cup \tilde S$: the support of such model is formed by the $\rightarrow$-normal forms of our ground terms with the obvious interpretation for the function and constant symbols. For simplicity, we assume that $\nullv$ is in normal form.~\footnote{To be pedantic, according to  the definition of $\Delta^+(\cM)$, there should be an equality $\nullv = c_0$ in $\Delta^+(\cM)$ so that $c_0$ is the normal form  of $\nullv$.} We need to check that whenever we have\footnote{We use $\rightarrow^*$ for the reflexive-transitive closure of $\rightarrow$ and $\rightarrow^+$ for the transitive closure of $\rightarrow$.} $f(t)\rightarrow^* \nullv$ then we have also $t\rightarrow^* \nullv$: we prove this by induction on the reduction ordering for our ground terms. 
 Let $t$ be a  term such that $f(t)\rightarrow^* \nullv$: if $t$ is $\ue$-free then the claim is trivial (because the axioms~\eqref{eq:null} are supposed to hold in $\cM$). Suppose also that induction hypothesis applies to all terms smaller than $t$. If $t$ is not in normal form, then let $\tilde t$ be its normal form; then we have $f(t)\rightarrow^+ f(\tilde t)\rightarrow^* \nullv$, by the fact that $\rightarrow$ is convergent. By induction hypothesis, $\tilde t\rightarrow \nullv$, hence $t\rightarrow^+ \tilde t\rightarrow^* \nullv$, as desired.
 %
% by induction hypothesis $\tilde t$ is $\ue$-free, hence $\tilde t$ is
%(the name of) some $c\in \vert \cM\vert$. But then $f(t)\rightarrow^* f(c)\rightarrow \nullv$ (by the convergence of $\rightarrow$), which means that $\cM\models f(c) = \nullv$ and so $\cM\models c=\nullv$; the latter (being $c$ in normal form as a name for an element of the support of $\cM$) means that $c\equiv \nullv$. In conclusion, $t\rightarrow^* \tilde t\equiv c \equiv \nullv$, as desired. 
Finally, let us consider the case in which $t$ is in normal form; since $f(t)$ is reducible in root position by some rule $l\to r$, our rules $l\to r$ are $\ue$-flat and $t$ is not $\ue$-free, we have that $t\equiv e_j$ for some existential variable $e_j$. Then, we must have that $S$ contains an equality of the kind $f(e_j)=u(\uy)\, \Vert\, D$ or of the kind $f(e_j)=e_i\, \Vert\, D$
(the constraint $D$ being true in $\cM$ under the given assignment to the $\uy$).
The latter case is reduced to the former,  since $e_i\rightarrow^* \nullv$ (by the convergence of $\rightarrow^*$) and since $S$ is closed under Paramodulation. In the former case, by 
the rule $Ext(\nullv)$, we must have that $S$ contains $e_j=\nullv\, \Vert\, D\cup \{u(\uy)= \nullv\}$. Now, since $f(e_j)=u(\uy)\, \Vert\, D$ belongs to $S$ and $D$ is true in $\cM$, we have that the normal forms of $f(e_j)$ and of $u(\uy)$ are the same; since the normal form of $f(e_j)$ is $\nullv$, the normal form of $u(\uy)$ is $\nullv$ too, which means that $u(\uy)=\nullv$ is true in $\cM$. But 
$e_j=\nullv\, \Vert\, D\cup \{u(\uy)= \nullv\}$ belongs to $S$, hence $e_j=\nullv$ belongs to $\tilde S$, which implies $e_j\rightarrow^* \nullv$, as desired.
$\hfill\dashv$

\subsection{Remarks on \textsc{MCMT} implementation}\label{subsec:mcmt}

%\vspace{-2mm}
As evident from  Subsection~\ref{sec:readonly}, our main motivation for investigating covers originated from the verification of data-aware processes. Such applications require database (DB) signatures to contain only unary function symbols (besides relations of every arity). We  observed that computing covers of primitive formulae in such signatures requires only polynomial time. In addition, if relation symbols are at most binary,
%
%Notice also that, in case of unary function symbols and of at most binary relation symbols, 
\emph{the cover of a primitive formula is a conjunction of literals}: this is crucial in applications, because model checkers like \textsc{mcmt}~\cite{mcmt} and \textsc{cubicle}~\cite{cubicle_cav} represent sets of reachable states as primitive formulae. 
%(`cubes', whence the name \textsc{cubicle}) and so they do not need any DNF conversion after computing covers in the above mentioned restricted case. 
%All this makes the approach via cover computations a quite attractive technique in database driven model checking.
This makes cover computations a quite attractive technique in database-driven model checking.

Our cover algorithm for DB signatures has been implemented in the model checker \textsc{mcmt}.
The implementation is however still partial, nevertheless the tool is able to compute covers 
for the $EUF$-fragment with unary function symbol, unary relations and binary relations. The optimized procedute of Section~\ref{app:complexity} has not yet been implemented, instead \textsc{mcmt} uses a customary Knuth-Bendix completion (in fact, for the above mentioned fragments constraints are always trivial and our constrained Superposition Calculus essentially boils down to Knuth-Bendix completion for ground literals in $EUF$.

Axioms~\eqref{eq:null}  are also covered in the following way. We assume that constraints of which we want to compute the cover always contains either 
         the literal $e_j=\nullv$ or the literal $e_j\neq \nullv$ for every existential variable $e_j$. Whenever a constraint contains the literal $e_j\neq \nullv$, the completion procedure adds the literal $u(y_i)\neq \nullv$ whenever it had produced a literal of the kind 
         $f(e_j)= u(y_i)$.\footnote{This is sound because $e\neq \nullv$ implies $f(e)\neq \nullv$ according to~\eqref{eq:null}, 
        so $u(y_i)\neq \nullv$ follows from $f(e_j)= u(y_i)$ and $e\neq \nullv$.}
        
        We wonder whether we are justified in assuming that all onstraints of which we want to compute the cover always contains either 
         the literal $e_j=\nullv$ or the literal $e_j\neq \nullv$ for every existential variable $e_j$. The answer is the affirmative: according to the backward search algorithm 
        implemented in array-based systems tools,
        the variable $e_j$ to be eliminated always comes from the guard of a transition and we can assume that such a guard contains  the literal  $e_j\neq \nullv$ (if we need a transition with $e_j= \nullv$ - for an existentially quantified variable $e_j$ - it is possible to write trivially this condition without using a quantified variable).  %AG: I'm not sure: if we want to generate new non-deterministic values it could be desirable to have e=\nullv, but the sentence seems correct anyway, with the proviso that we split the transitions
        The \textsc{mcmt} User Manual (available from the distribution) contains precise instructions on how to write specifications following the above prescriptions.

A first experimental evaluation (based on the existing benchmark provided in \cite{verifas}, which samples 32 real-world BPMN workflows taken from the BPMN official website \url{http://www.bpmn.org/}) is described in~\cite{CGGMR18}.
The benchmark set is available as part of the last distribution~2.8 of
\textsc{mcmt}
%
%\centerline{
\url{http://users.mat.unimi.it/users/ghilardi/mcmt/}
%}
%
%\noindent
 (see the subdirectory \texttt{/examples/dbdriven} of the distribution). The User Manual, also included in the distribution,
 contains a dedicated section giving essential information on how to encode relational artifact systems (comprising \emph{both first order and second order variables}) in \textsc{mcmt}
 specifications and how to produce
 user-defined examples in the database driven framework. 
 The first experiments were very encouraging: the tool was able to solve in few seconds all the proposed benchmarks  and the cover computations generated automatically during model-checking search were discharged instantaneously, see~\cite{CGGMR18}  for more information about our experiments.

 \section{Conclusions and Future Work}\label{sec:conclusions}

 The above experimental setup motivates new research %\marginpar{espandere queste conclusioni...} 
 to extend Proposition~\ref{prop:mc} to further theories axiomatizing integrity constraints used in DB applications. Combined cover algorithms (along the perspectives in~\cite{GM})  could be crucial also in this setting. %for database driven verification.  
 Practical algorithms for the computation of covers in the theories falling under the hypotheses of Proposition~\ref{prop:mc} need to be designed: as a little first example, in Subsection~\ref{subsec:extension} above we showed how to handle Axiom~\eqref{eq:null} by light modifications to our techniques. Symbol elimination of function and predicate variables should also be combined with cover computations.
 
 We consider the present work, together with \cite{CGGMR18,CGGMR19}, as the starting point for a full line of research
dedicated to SMT-based techniques for the effective verification of data-aware
processes, addressing richer forms of verification beyond safety (such as
liveness, fairness, or full LTL-FO) and richer classes of artifact systems,
(e.g., with concrete data types and arithmetics), while identifying novel
decidable classes (e.g., by restricting the structure of the DB and of
transition and state formulae) beyond the ones presented in~\cite{CGGMR18,CGGMR19}.
Concerning implementation, we plan to further develop our tool to incorporate
in it
%Implementation-wise, we want to build on the reported encouraging results and
%benchmark our approach using the \textsc{Verifas} system as a baseline, while
%incorporating
the plethora of optimizations and sophisticated search strategies available in
infinite-state SMT-based model
checking.  Finally, in~\cite{BPM19} we tackle more conventional process modeling
notations, concerning in particular data-aware extensions of the de-facto
standard BPMN\footnote{\url{http://www.bpmn.org/}}: we plan to provide a full-automated translator from the data-aware BPMN model presented in~\cite{BPM19} to the artifact systems setting of~\cite{CGGMR19}.

%\hfill\eject
\bibliographystyle{abbrv}
\bibliography{mcmt}

\end{document}